\documentclass[reqno,a4paper]{amsart}
\usepackage{geometry} 
\usepackage{fullpage}          
\usepackage{amsmath}
\usepackage[usenames,dvipsnames]{pstricks}
\usepackage{auto-pst-pdf}
\usepackage{epsfig}
\usepackage{amsthm}
\newtheorem{thm}{Theorem}[section]
\newtheorem{lemma}[thm]{Lemma}
\newtheorem{definition}[thm]{Definition}
\newtheorem{corollary}[thm]{Corollary}
\newtheorem{remark}[thm]{Remark}
\newtheorem{proposition}[thm]{Proposition}
\usepackage{graphicx}
\usepackage{amssymb}
\usepackage{epstopdf}
\usepackage[all]{xy}
\usepackage{hyperref}
 \usepackage[usenames,dvipsnames]{pstricks}
 \usepackage{epsfig}
\numberwithin{equation}{section}

\newcommand\Ptw{\hbox{P}_{\scriptsize
       \hbox{II}}}
\newcommand\Pth{\hbox{P}_{\scriptsize
       \hbox{34}}}

\begin{document}

\title[Asymptotics of $\Ptw$]{Global Asymptotics of the Second Painlev\'e equation in Okamoto's space}
\author{P. Howes}\thanks{This research was supported by an Australian Postgraduate Award. }
\address{School of Mathematics and Statistics F07, The University of Sydney, NSW 2006, Australia\\ Fax: +61 2 9351 4534}
\email{howespt@gmail.com}
\author{N. Joshi}\thanks{This research was supported by the Australian Research Council grant \# DP110102001. }
\address{School of Mathematics and Statistics F07, The University of Sydney, NSW 2006, Australia\\ Tel: +61 2 9351 2172\\ Fax: +61 2 9351 4534}
\email{nalini.joshi@sydney.edu.au}
\begin{abstract}
We study the solutions of the second Painlev\'e equation ($\Ptw$) in the space of initial conditions first constructed by Okamoto, in the limit as the independent variable, $x$, goes to infinity. Simultaneously, we study solutions of the related equation known as the thirty-fourth Painlev\'e equation ($\Pth$). By considering degenerate cases of the autonomous flow, we recover the known special solutions, which are either rational functions or expressible in terms of Airy functions. We show that the solutions that do not vanish at infinity possess an infinite number of poles. An essential element of our construction is the proof that the union of exceptional lines is a repellor for the dynamics in Okamoto's space. Moreover, we show that the limit set of the solutions exists and is compact and connected.
\end{abstract}
\keywords{The second Painlev\'e equation, thirty-fourth Painlev\'e equation, asymptotic behaviour, resolution of singularities, rational surface}
\subjclass[2000]{34M55; 34E05,34M55, 34M30,14E15}

\date{}
\maketitle

\section{Introduction}
The second Painlev\'e equation 
\begin{align}\label{painleve2}
\frac{d^2y}{dx^2}&=2y^3+xy+\alpha,
\end{align}
is well known from its application in random matrix theory \cite{tw:94,m:04} amongst many others. For integer values of $\alpha$, it is known that there exist rational solutions, which are expressed as the logarithmic derivative of ratios of successive Yablonskii-Vorob'ev polynomials.  In this paper, we study the asymptotic behaviour of all solutions in complex projective space of dimension two for $x\in{\mathbb C}$, with $|x|\to\infty$. Throughout the paper we assume that $\alpha$ is fixed and bounded.

The asymptotic study of $\Ptw$ was initated by Boutroux \cite{bou:14} and has more recently been studied in \cite{joshi:88,joshi:92,kit:94,dei:95,fok:06}. Boutroux provided a change of variables in which the asymptotic behaviours become explicit and studied the equation directly, whilst more recent approaches have centred around the Riemann-Hilbert method. This paper uses the explicit construction of the blown up space of initial conditions in order to give a description of the solution space as the modulus of the independent variable approaches infinity.

It is known  through Okamoto's \cite{oka:79} compactification of the space of initial conditions that the solution space is connected. However, until the work of Duistermaat and Joshi \cite{djo:10} in the case of ${\hbox{P}_{\scriptsize\hbox{I}}}$, no completeness and connectedness study of the asymptotic behaviours had been carried out for any of the Painlev\'e equations to our knowledge. In this paper we fill this gap for the second and thirty-fourth Painlev\'e equations.

The main result of this paper is to show compactness and connectedness of the limit sets of solutions to the second and thirty-fourth \cite{inc:56} Painlev\'e equations as the independent variable approaches infinity. A more precise statement is made in Section \ref{statements}. The proof of this statement relies on the construction of a function $d$, which is used to measure the distance of a solution of (\ref{painleve2}) to a set of exceptional lines created in the process of blowing up (\ref{painleve2}) in $\mathbb{P}^2(\mathbb{C})$. As a corollary, we find that all solutions to $\Ptw$ which are not uniformly asymptotically zero must necessarily have infinitely many poles.

The paper is organised as follows: In Section \ref{setup}, we provide a change of variables to make more explicit the asymptotic behaviour of (\ref{painleve2}), as well as provide some definitions and explain the terminology to be used in subsequent sections. 
In Section \ref{statements}, we introduce the main result, as well as provide some corollaries to this result.
In Section \ref{specsols}, we show how the special solutions to (\ref{painleve2}) arise naturally by considering the degenerate values of the autonomous energy function described in Section \ref{setup}. Furthermore, we will make some remarks about the resolution of singularities for nonlinear systems, and their dependence on choices of coordinates and asymptotic limits. The blow up (resolution of singularities) for the system (\ref{boutrouxsys})-(\ref{boutrouxsys2}) is carried out in Appendix \ref{app1}, where are the details are given explicitly as the proofs of the statements in Section \ref{statements} require precise asymptotic estimates. In Section \ref{conc} we make some concluding remarks.


\section{Hamiltonian System and Definitions}\label{setup}

The second Painlev\'e equation (\ref{painleve2}) has a Lagrangian $L$ given by

\begin{equation}
\label{p2}
L(Q,\dot{Q},x)=\dfrac{\dot{Q}^2}{2}+\dfrac{Q^4}{2}+x\dfrac{Q^2}{2}+\alpha Q,
\end{equation}
which with the Euler-Lagrange equations implies (\ref{painleve2}) for $Q$. Performing a standard Legendre transformation leads to an associated Hamiltonian
\begin{equation}
\label{preham}
H(Q,P,x)=\dfrac{P^2}{2}-\dfrac{Q^4}{2}-x\dfrac{Q^2}{2}-\alpha Q.
\end{equation}
Whilst this Hamiltonian system implies (\ref{painleve2}) for Q, we make the canonical change of variables with type 3 generating function, $G_3=-pQ+Q^3/3+xQ/2$:
\begin{subequations}
\begin{align}
\label{half}
   P &= p-(Q^2+x/2),    \\
   q &= Q.
\end{align}
\end{subequations}
Under this change the Hamiltonian becomes

\begin{equation}\label{hamiltonian1}
H(q,p;x)=\dfrac{p^2}{2}-(q^2+x/2)p-(\alpha+1/2)q,
\end{equation}
which with Hamilton's equations implies the second Painlev\'e equation (\ref{painleve2}) for $q$ \textit{and} the thirty-fourth Painlev\'e equation for $p$:
\begin{equation}
\label{p34}
\dfrac{d^2p}{dx^2}=\dfrac{1}{2p}\left(\dfrac{dp}{dx}\right)^2+2p^2-xp-\dfrac{\beta}{2p},
\end{equation}
where $\beta=(\alpha-1/2)^2$.

Thus from this point forward, relevant statements made in the analysis of $\Ptw$ hold analogously for $\Pth$. Further benefits of choosing this coordinatisation will be discussed in Subsection \ref{coords}.

\subsection{Boutroux Scaling}
We are interested in studying the limit $x\rightarrow \infty$, $x\in \mathbb{C}$ and so we perform another change of variables (\`a la Boutroux \cite{bou:14}) to make the leading order asymptotic behaviour explicit. The Hamiltonian (\ref{hamiltonian1}) is almost weighted homogeneous, which inspires the following change of variables:
 \begin{eqnarray*}
q & = & \lambda u, \\
p & = & \lambda^2 v,\\
x &=& \lambda^2 \zeta ,
\end{eqnarray*}
and then 
\begin{equation}
H=\lambda^4\left(\dfrac{v^2}{2}-(u^2v+\dfrac{\zeta v}{2}-\dfrac{(\alpha+1/2)u}{\lambda^3})\right).
\end{equation}
$\zeta=1$ if and only if $\lambda^2=x$. Then if $x=x(z)$ and $\dot{}$ is differentiation with respect to $z$, then 
 
\begin{displaymath}
 \dot{u}=\dot{x}\left((v-u^2-1/2)x^{1/2}-\dfrac{u}{2x}\right),
\end{displaymath}
and so if we choose $\dot{x}x^{1/2}=1$, then we have $z=2/3x^{3/2}$. We then have the system:
 \begin{eqnarray}\label{boutrouxsys}
\dot{u} & = & v-u^2-1/2-\dfrac{u}{3z},  \\ \label{boutrouxsys2}
\dot{v} & = & 2uv+\dfrac{2\alpha+1}{3z}-\dfrac{2v}{3z} .
\end{eqnarray}
This pair of equations implies the Boutroux forms of $\Ptw$ and $\Pth$ for $u$ and $v$ respectively:
\begin{subequations}
\begin{align}
\label{bp2}
 \ddot{u}&=2u^3+u+\dfrac{2\alpha}{3z}-\dot{u}\dfrac{1}{z}+\dfrac{u}{9z^2} \\ 
  \ddot{v}&=\dfrac{\dot{v}^2}{2v}+2v^2-v-\dfrac{\dot{v}}{z}-\dfrac{(2\alpha-1)^2}{18vz^2}+\dfrac{1}{4z^2}.  \label{bp32}  
\end{align}
\end{subequations}
The Boutroux system (\ref{boutrouxsys})-(\ref{boutrouxsys2}) is an order $z^{-1}$ perturbation of the autonomous system with time independent Hamiltonian given by
\begin{equation}\label{energyfct}
E:=v^2/2-u^2v-v/2,
\end{equation}
where 

\begin{equation}
\dot{E}:=\dfrac{-4E}{3z}+\dfrac{4\alpha v-(2\alpha+1)(2u^2+1)}{6z}\,\,.
\end{equation}

\begin{remark}
The second Painlev\'e equation has the B\"acklund transformations

\begin{equation}
\label{bt1}
\mathcal{T}^{{\pm}}:\enskip y(x;\alpha\pm 1)=-y-\dfrac{2\alpha\pm 1}{2y^{2}\pm 2y^{{\prime}}+x}.
\end{equation}
Under the Boutroux change of variables, the B\"acklund transformation becomes
\begin{equation}
\label{bt2}
\mathcal{T}^{{\pm}}:\enskip u(z;\alpha\pm 1)=-u-\dfrac{2\alpha\pm 1}{3zu^{2}\pm u\pm 3z\dot{u}+3z/2}.
\end{equation}
Upon elimination of the derivative term in (\ref{bt1}), one finds the equation referred to as $alt-dP_I$. In the Boutroux setting, the equation becomes a Boutroux form of $alt-dP_I$:

\begin{equation}
\label{baltdp1}
\dfrac{2\alpha+1}{\bar{u}+u}+\dfrac{2\alpha-1}{\underbar{u}+u}+6zu^2+3z=0,
\end{equation}
where $u=u(z,\alpha)$ and $\bar{u}=u(z,\alpha+1)$, $\underbar{u}=u(z,\alpha-1)$. This is an appropriate scaling for the large $z$ limit for $alt-dP_I$.

\end{remark}

\subsection{Notation and Definitions}
The natural setting to study the Painlev\'e equations is in complex projective space, where the poles become zeroes in the coordinate charts near infinity in the affine plane. Following the pioneering work of Okamoto \cite{oka:79}, we study the Painlev\'e system in $\mathbb{P}^2(\mathbb{C})$. 
We embed the Boutroux-Painlev\'e system in $\mathbb{P}^2(\mathbb{C})$ and identify the affine coordinates with homogeneous coordinates as
\begin{eqnarray*}
[1:u:v] &= &[u^{-1}:1:vu^{-1}]=[u_{01}:1:v_{01}]\\
&=&[v^{-1}:uv^{-1}:1]=[v_{02}:u_{02}:1]\,.
\end{eqnarray*}

It is known from \cite{oka:79} that every continuous Painlev\'e equation can be regularised (made free from the indeterminacy of the flow through base points) by blowing up $\mathbb{P}^2(\mathbb{C})$ at 9 points. Indeed it is known from Sakai's classification \cite{sak:01} that all (continuous and discrete) Painlev\'e equations can be regularised by a 9-point blow up of $\mathbb{P}^2(\mathbb{C})$.

We denote the line at infinity by $L_0$. Note that it is given by $u_{01}=0$ or $v_{02}=0$. For $0\leq i\leq 8$ corresponding to the $i$-th stage of the blowing up process, we denote the $i$-th base point $b_i$ and the exceptional line attached to the base point by $L_{i+1}$ and the coordinates of the two charts coordinatising $L_{i+1}$ by $(u_{ij},v_{ij})$, $j=1,2$. 

Moreover, we denote the proper transform of $L_i$ after the ninth blow up as $L_i^{(9-i)}$ and the $z$-dependent set of lines where the vector field becomes infinite by $I(z):=\bigcup_{i=0, i\neq 6}^{8} L_{i}^{(9-i)}(z)$. This set will be referred to as \textit{the infinity set}. The final space obtained by regularising the Boutroux-Painlev\'e system is denoted $S_9(z)$, but we will often drop the explicit $z$ dependence for ease of notation. 

In the proofs of the results, we regularly make use of the Jacobian of the coordinate transformation from $(u,v)$ to the new coordinate system, given by 
\begin{displaymath}
w_{ij} = \dfrac{\partial u_{ij}}{\partial u}\dfrac{\partial v_{ij}}{\partial v}-\dfrac{\partial u_{ij}}{\partial v}\dfrac{\partial v_{ij}}{\partial u} .
\end{displaymath}

\section{Statement of results} \label{statements}
 
In this section we will state the main result: that the limit set of the solutions of the second and thirty-fourth Painlev\'e equations forms a non-empty, compact and connected subset of $S_9(\infty) \backslash I(\infty)$, which is the space $S_9(z) \backslash I(z)$ obtained by replacing all $1/z$ terms by $0$. The result is a consequence of Theorem \ref{theorepel}, which will be proved in Section \ref{proof}. In this theorem, $d(z)$ denotes a distance measure to the infinity set $I(z)$, that is, $d(z)=0$ if and only if the solution to the Painlev\'e system crosses $I(z)$. It is shown in Lemma \ref{lemma2} that a continuous such distance measure exists. The corollaries which follow from Theorem \ref{theorepel} are proved in this section. 

\begin{thm}\label{theorepel}

Let $\epsilon_1,\epsilon_2,\epsilon_3$ be given such that $\epsilon_1>0$, $0<\epsilon_2<4/3$, $0<\epsilon_3<1$. Then there exists $\delta\in\mathbb{R}_{>0}$ such that if $|z_0|>\epsilon_1$, and $|d(z_0)|<\delta$, then

\begin{displaymath}
 \rho = \sup \left\{r>|z_0|\,\, \mbox{such that}\,\, |d(z)|<\delta\,\, \mbox{whenever}\, |z_0|\leq |z|\leq r\right\}
\end{displaymath} satisfies

\begin{description}
  \item[\textrm{(i)}]
  \begin{displaymath}
\delta \geq |d(z_0)| \left(\frac{\rho}{z_0}\right)^{4/3-\epsilon_2}(1-\epsilon_3)\,.
\end{displaymath}
  \item[\textrm{(ii)}] If $|z_0|\leq |z| \leq \rho$, then 
  \begin{displaymath}
d(z)=d(z_0)\left(\frac{z}{z_0}\right)^{4/3+\epsilon_2(z)}(1+\epsilon_3(z))\, .
\end{displaymath}
  \item[\textrm{(iii)}] If $|z|\geq \rho$ then
  \begin{displaymath}
|d(z)|\geq \delta (1-\epsilon_3) \, .
\end{displaymath}
\end{description}

\end{thm}

\begin{remark} Note that in case (iii) of Theorem \ref{theorepel}, we have $d(z)>0$ for all complex times $z$, $|z|\geq|z_0|$ and so the solution never reaches the infinity set.
\end{remark}

The following is a definition of a limit set for complex dynamical systems.

\begin{definition}
For every solution $\mathbb{C} \backslash \{0\} \ni z \mapsto U(z) \in S_9(z) \backslash I(z)$, let $\Omega_U$ denote the set of all $s \in S_9(\infty) \backslash I(\infty)$ such that there exists a sequence $z_j \in \mathbb{C}$ with the property that $z_j \rightarrow \infty$ and $U(z_j) \rightarrow s$ as $j \rightarrow \infty$. The subset $\Omega_U$ of $S_9(\infty) \backslash I(\infty)$ is called the limit set of the solution $U$.
\end{definition}

\begin{thm}\label{corlimitset}

There exists a compact subset $K$ of $S_9(\infty) \backslash I(\infty)$ such that for every solution $U$ the limit set $\Omega_U$ is contained in K. The limit set $\Omega_U$ is a non-empty, compact and connected subset of K, invariant under the flow of the autonomous Hamiltonian system on $S_9(\infty) \backslash I(\infty)$. For every neighbourhood A of $\Omega_U$ in $S_9$ there exists an $r>0$ such that $U(z)\in A$ for every $z\in\mathbb{C}$ such that $|z|>r$. If $z_j$ is any sequence in $\mathbb{C}\backslash \{0\}$ such that $z_j \rightarrow \infty$ as $j\rightarrow \infty$, then there is a subsequence $j=j(k)\rightarrow\infty$ as $k\rightarrow \infty$ and an $s\in \Omega_U$ such that $U(z_{j(k)})\rightarrow s$ as $k\rightarrow \infty$.
\begin{proof}

See Corollary 4.6 in \cite{djo:10}, whose method of proof also applies here.

\end{proof}

\end{thm}

\begin{lemma}\label{invariance}

$\Omega_U$ is invariant under the transformation $(E,u,v)\mapsto (E,-u,v)$.

\begin{proof}

It is known that the solutions $y(x)$ to the second Painlev\'e equation are single valued in the complex plane. The transformation to Boutroux coordinates $y(x)=x^{1/2}u(z)$, $z=2x^{3/2}/3$ is singular at $z=0$ and correspondingly $x=0$. They introduce multivaluedness of the solutions $u(z)$ when $z$ runs around the origin. The single valuedness of $y(x)$ and the relation $u(z)=(3z/2)^{-1/3} y((3z/2)^{2/3})$ implies that the analytic continuation of $u(z)$ along the path $z e^{i\theta}$ becomes $-u(z)$ when $\theta$ runs from 0 to $3\pi$. That is 

\begin{align}
&u(ze^{3\pi i})=-u(z),
\end{align}
and from \eqref{boutrouxsys} we have
\begin{align}
&v(ze^{3\pi i}) =v(z).
\end{align}
Recall also from \eqref{energyfct} that $E=v^2/2-u^2v-v/2$.  The lemma follows from these results.

\end{proof}

\end{lemma}

\begin{corollary}\label{infpoles}
Every solution $u(z)$ of the second Painlev\'e equation whose limit set is not $\{0\}$ has infinitely many poles.

\begin{proof}
Let $u(z)$ be a solution of the Boutroux-Painlev\'e equation with only finitely many poles. Let $U(z)$ be the corresponding solution of the system in $S_9 \backslash S_{9}(\infty)$, and $\Omega_U$ the limit set of $U$. 

According to Theorem \ref{corlimitset}, $\Omega_U$ is a compact subset of $ S_{9}(\infty)\backslash I(\infty)$. If $\Omega_U$ intersects one of the two pole lines $L_6^{(3)}(\infty)$ or $L_9(\infty)$ at a point $p$, then there exists $z$ with $|z|$ arbitrarily large such that $U(z)$ is arbitrarily close to $p$, when the transversality of the vector field to the pole lines implies that $U(\zeta) \in (L_6^{(3)}\cup L_9)$ for a unique $\zeta$ near $z$, which means that $u(z)$ has a pole at $z = \zeta$. 

As this would imply that $u(z)$ has infinitely many poles, it follows that $\Omega_U$ is a compact subset of $S_9 (\infty) \backslash (I(\infty)\cup L_6^{(3)}(\infty) \cup L_9(\infty))$. However, $L_6^{(3)}(\infty) \cup L_9(\infty)$ is equal to the set of all points in $S_9(\infty)\backslash I(\infty)$ which project to the line $L_0(\infty)$ in the complex projective plane, and therefore $S_9 (\infty)\backslash (I(\infty)\cup L_6^{(3)}(\infty) \cup L_9(\infty))$ is the affine $(u, v)$ coordinate chart, of which $\Omega_U$ is a compact subset, which implies that $u(z)$ and $v(z)$ remain bounded for large $|z|$. 

It follows from the boundedness of $u$ and $v$ that $u(z)$ and $v(z)$ are equal to holomorphic functions of $1/z$ in a neighbourhood of $z=\infty$, which in turn implies that there are complex numbers $u(\infty)$, $v(\infty)$ which are the limit points of $u(z)$ and $v(z)$ as $|z| \rightarrow \infty$. In other words, $\Omega_U = \{(u(\infty), v(\infty))\}$, a one point set. Because the limit set $\Omega_U$ is invariant under the autonomous Hamiltonian system and contains only one point, this point is an equilibrium point of the autonomous Hamiltonian system given by:
\begin{subequations}
\begin{align}
\dot{u}=& v- u^2-1/2 \\
\dot{v}=& 2uv\, .
\end{align}
\end{subequations}
That is, $(u(\infty),v(\infty))\in \{(0,1/2),(i/\sqrt{2},0),(-i/\sqrt{2},0)\}$. By assumption in the corollary, we assume $(u(\infty),v(\infty)) \neq (0,1/2)$, and so $u(\infty)$ must equal one of $\pm i/\sqrt{2}$. But according to lemma \ref{invariance}, $(a,b)\in \Omega_U \iff (-a,b)\in\Omega_U$. This contradicts the deduction that $\Omega_U$ is a one point set.

\end{proof}

\end{corollary}

\section{Special Solutions to $\Ptw$ in the Okamoto-Boutroux space}\label{specsols}

For generic values of $E$, Equation  (\ref{energyfct}) defines an elliptic curve with four distinct branch points. A natural question to ask is what happens at degenerate values of the autonomous Hamiltonian energy function $E$. Define 
\begin{equation}
\label{degH}
h(u,v,E):=v\left(\dfrac{v}{2}-\dfrac{1}{2}-u^2\right)-E 
\end{equation}
For generic $E$ the roots of $h$ are distinct. However, the roots of $h$ are no longer distinct when its gradient vanishes. These \lq\lq singular\rq\rq\ points are given by
\begin{equation}
\label{dis}
\nabla h=(h_u, h_v)  = \left( 2v(v^2-v-2E), u^4+u^2+\dfrac{1}{4}+2E\right)=0
\end{equation}
Simultaneous with $h=0$, these equations imply that $E=0$ or $E=-1/8$. In the case of $E=0$, Equation (\ref{energyfct}) gives
\begin{subequations}
\begin{align}
 (u,v)  &=(u,0)  \label{ecase1}\\
\noalign{or}
  (u,v)  &=(u,2u^2+1 )\label{ecase2}
\end{align}
\end{subequations}
In the case (\ref{ecase1}), then
\begin{align*}
\dot{E}&=\dfrac{-4E}{3z}+\dfrac{4\alpha v-(2\alpha+1)(2u^2+1)}{6z} \\
&=0+\dfrac{(2\alpha+1)(2u^2+1)}{6z}
\end{align*}
and so if $\alpha=-1/2$ then $\dot{E}=0$. An inductive argument on the $n$-th derivative of $E$ gives the following result:

\begin{proposition}
Under the conditions $v=0$ and $\alpha=-1/2$, $E^{(n)}=0$.
\end{proposition}

These conditions are exactly those needed to specify the seed solution from which the family of Airy type solutions are generated. The Airy family of solutions for (\ref{painleve2}) is given as the logarithmic derivative of the general solution to Airy's differential equation in the form
\begin{equation}
\label{airyeqn}
\dfrac{d^2y}{dx^2}+\dfrac{x}{2}y=0.
\end{equation}
Under Boutroux's change of variables, the governing Airy equation becomes 
\begin{equation}
\label{airyeqnboutroux}
\dfrac{d^2y}{dz^2}+\dfrac{dy}{dz}\dfrac{1}{3z}+\dfrac{y}{2}=0,
\end{equation}
where $u(z)$, the solution to the Boutroux form of $\Ptw$ is related to $y(z)$ by $u(z)=\dot{y}/y$. The formal solution to this equation near a regular point $z_0$ is given by
\begin{align*}
   y(z) &=\sum_{k=0}^{\infty} a_k (z-z_0)^k.   
\end{align*}
The first two coefficients $a_0$ and $a_1$ are free, whilst the remaining $a_k$ satisfy 
\begin{align}
\nonumber   &a_2 =\dfrac{a_0}{4}+\dfrac{a_1}{6z_0},\\
\label{rec}   &3(k+2)(k+3)a_{k+3}+(k+2)(3k+4)a_{k+2}+\dfrac{3z_0}{2}a_{k+1}+\dfrac{3}{2}a_k=0, \,\, k\ge0 \, . 
\end{align}

A zero of the Airy function at $z=z_0$ corresponds to a pole of $u(z)$, with Laurent expansion given by
\begin{align}
\nonumber
 u(z) = &\frac{1}{z-z_0}+\frac{1}{6 z_0}+\frac{\left(-18 z_0^2-19\right) (z-z_0)}{108
   z_0^2}+\frac{\left(9 z_0^2+37\right) (z-z_0)^2}{216 z_0^3}\\
\label{lau_airy}   &\qquad+\frac{\left(-324
   z_0^4-2520 z_0^2-9565\right) (z-z_0)^3}{58320 z_0^4}+{\mathcal O}\left((z-z_0)^4\right).
\end{align}

In this case we have the solution crossing the pole line defined by $u_{92}=0$. In fact, there is a one-to-one correspondence between the Airy type solutions and the solutions which have a pole at $(u_{92},v_{92})=(0,0)$, when $\alpha=-1/2$. That is to say the Airy solutions are the unique solutions which cross the pole line $u_{92}=0$ at $v_{92}=0$ for $\alpha=-1/2$. 

The second zero (\ref{ecase2}) of the Hamiltonian energy is also interesting. In this case we find the condition $v=2u^2+1$ is consistent with the equations (\ref{boutrouxsys})-(\ref{boutrouxsys2}) if and only if $\alpha=1/2$. By considering the $z$-derivative of $E$, we find
\begin{align}
\dot{E}&=\dfrac{(2\alpha -1)(2u^2+1)}{6z}.
\end{align}
We have here two new cases for a zero derivative. For $\alpha=1/2$, we have Airy solutions as before, related by the B\"acklund transformation (\ref{bt2}) which shifts $\alpha$ to $\alpha+1$. However, in the case of $u(z)\sim\pm i/\sqrt{2}+O(z^{-1})$, we have a different type of solution. This is the leading order behaviour of the tronqu\'ee type solutions to $\Ptw$.

The other singular value of the autonomous energy $E$ which is of interest is when $E=-1/8$. In this case two of the four distinct branch points of the elliptic curve coalesce and the curve is degenerate. The double point occurs at $u=0$, $v=1/2$. This is the leading order behaviour of the rational solutions to the system (\ref{boutrouxsys})-(\ref{boutrouxsys2}). As seen in Corollary \ref{infpoles}, besides this degeneracy, all other solutions have infinitely many poles. We see that the special solutions of the nonautonomous equation correspond to degenerations in the autonomous flow of the system near infinity. 

\subsection{Blowing-up and coordinate systems}\label{coords}
In Section \ref{setup}, a coordinate change was performed to bring the Hamiltonian system for $\Ptw$ to the form (\ref{hamiltonian1}). We could have continued with the system of equations given by Hamilton's equations of motion for the system (\ref{preham}). However, there are several reasons we preferred the system (\ref{hamiltonian1}) over the former. Firstly, it allows us to study both $\Ptw$ and $\Pth$ in tandem. Secondly, it allows the reader to tie in the results with that of Okamoto \cite{oka:79}, who also used this coordinate choice. A third consideration is one which seems to be overlooked in the literature, which is that the number of blowups and the location of the base-points depends on the choice of coordinate system. If one were to study the system defined by (\ref{preham}), we would find the following sequence of base points:
\begin{align}
\nonumber (u_{02},v_{02})&= (u/v,1/v)\xleftarrow{(0,0)} (u_{12},v_{12})\xleftarrow{(0,0)} (u_{21},v_{21})    \\
\label{seq}
   &\xleftarrow{(1,0)} (u_{31},v_{31}) \xleftarrow{(0,0)} (u_{41},v_{41}) \xleftarrow{(-1/2,0)}  (u_{51},v_{51}) \xleftarrow{(\frac{1-2\alpha}{3z},0)}  (u_{61},v_{61}),\\
 \nonumber  &(u_{21},v_{21})   
   \xleftarrow{(-1,0)} (u_{71},v_{71}) \xleftarrow{(0,0)} (u_{81},v_{81}) \xleftarrow{(1/2,0)}  (u_{91},v_{91}) \xleftarrow{(\frac{1+2\alpha}{3z},0)} (u_{101},v_{101}).
\end{align}
where the label on each arrow represents the centre of the blow up in the preceding coordinate chart. From this one would be naively led to believe that 10 blow ups is the number required to resolve the singularities of this system. We can however bring this number back down to 9 if we consider a blow down of a line as a 'negative' blow-up. In this case, we can blow down the proper transform of the line $u_{01}=0$, a curve which can be considered the blow up of a regular point (having been a curve with self-intersection $-1$). Hence we see that the minimal blow up required is indeed canonical (in this case requiring 9 blow ups). Thus we observe that a 'good' coordinate choice in the beginning makes the calculation simpler and clearer.

\begin{remark}
The process of resolution presented explicitly in the appendix shows that the structure of the blow up (the number and location of the base points and the coordinate charts) of the Boutroux form of the Painlev\'e system is asymptotically close to that of the autonomous limit system, i.e., the system whereby the powers of $1/z$ in (\ref{boutrouxsys})-(\ref{boutrouxsys2}) are replaced with zero, which is solved by elliptic functions. Indeed, the difference between the two systems is only visible in the last two branches of the blow up, at the pole lines. 
This is noteworthy because the operations of blowing up and taking limits do not commute in general.
\end{remark}

\section{Proof of Theorem \ref{theorepel}}\label{proof}

In this section we show that the infinity set $I$ is a repeller. This implies that every solution which starts in Okamoto's space of initial conditions remains there for all complex times $z$. The proof of this result requires the construction of a distance measure $d$, used to measure the distance to the infinity set.

Let $\mathcal{S}$ denote the fiber bundle of the surfaces $S_9 = S_9(z)$, $z \in\mathbb{C}\backslash 0$. If $\mathcal{I}$ denotes the union in $\mathcal{S}$ of all $I(z)$, $z \in\mathbb{C} \backslash 0$, then $\mathcal{S}$ $\backslash$ $\mathcal{I}$ is Okamoto's Òspace of initial conditionsÓ, fibered by the surfaces $S_9(z) \backslash I(z)$. We will analyse the asymptotic behaviour for $|z| \rightarrow \infty$ of the solutions of the Painlev\'e system (\ref{boutrouxsys})-(\ref{boutrouxsys2}), by studying the $z$-dependent vector field in the coordinate systems introduced in Section \ref{setup}. 

Near the part of the infinity set given by $L_0^{(9)} \cup L_1^{(8)}\cup L_2^{(7)} \cup L_3^{(6)}\cup L_4^{(5)} \cup L_7^{(2)}$ we use the function $1/E$ to measure the distance to the infinity set. Because the lines $L_5^{(4)}$ and $L_8^{(1)}$ contain the lift of base points of $E$, the function $1/E$ is no longer a good measure of distance to the infinity set. These require an alternative measure of distance. Near $L_5^{(4)}$ we use $w_{62}$ and near $L_8^{(1)}$ we use $w_{91}$, where $w_{ij}$ is the Jacobian of the coordinate change from $(u,v)$ to $(u_{ij},v_{ij})$.
\begin{lemma}
The reciprocal of the autonomous Hamiltonian energy function E, and the Jacobians $w_{62}$ and $w_{91}$ are zero on the infinity set.
\end{lemma}
\begin{proof}
The calculations of $1/E$, $w_{62}$ and $w_{91}$ in each coordinate chart of Okamoto's space are provided in Appendix \ref{app1}. For each $0\le i \le 8$, the lines $L_i^{(9-i)}$ in the infinity set are determined by the relation $v_{i1}=0$, with $u_{i1}\in \mathbb{C}$. These show that in the limit as we approach an exceptional line $L_i^{(9-i)}$, we have 
\begin{equation}
\label{inf}
\left.
\begin{array}{ll}
   1/E &={\mathcal O}\bigl( v_{i1}^{\lambda_i}\bigr)\\
    w_{62} &={\mathcal O}\bigl(v_{i1}^{\mu_i}\bigr) \\  
    w_{91} &={\mathcal O}\bigl( v_{i1}^{\nu_i}\bigr) 
    \end{array}
    \right\}  \quad {\rm as}\, v_{i1} \rightarrow 0 \,
\end{equation}
for some some positive integers $\lambda_i, \mu_i, \nu_i$. It follows that all three functions vanish on each $L_i^{(9-i)}\in I$.
\end{proof}
\begin{lemma}\label{lemma1}
Let $I_3:= \cup_{i=0}^3 I_i^{(9-i)}$. For every $\epsilon > 0$, there exists a neighbourhood $U$ of $I^3$ in S such that $|\dot{E}/(E)+4/(3z)|<\epsilon$ in $U$ and for all $z\in\mathbb{C}$\textbackslash $\{0\}$. For every compact subset K of $L_4^{(5)}\backslash L_5^{(4)} \cup L_7^{(2)}\backslash L_8^{(1)}$ there exists a neighbourhood V of K in $S_9$ and a constant $C>0$ such that $|3z\dot{E}/(4E)|$ $\leq C$ in V and for all $z\in \mathbb{C}\backslash \{ 0\}$.
\end{lemma}
\begin{proof}
Because $I_3$ is compact, it suffices to show that every point of $I_3$ has a neighbourhood in which the estimate holds. On $I_3$, the function $r:=\dot{E}/(E)+4/(3z)$ is equal to the following in each coordinate chart:

\begin{align}
\label{firstr}
  r_{01} = & \dfrac{u_{01} \left(2 \alpha u_{01}^{2}-4 \alpha u_{01} v_{01}+4 \alpha+u_{01}^{2}+2\right)
}{3 v_{01} z \left(u_{01}^{2}-u_{01} v_{01}+2\right)}   \\
  r_{02} = &  \dfrac{v_{02} \left(4 \alpha u_{02}^{2}+2 \alpha v_{02}^{2}-4 \alpha v_{02}+2 u_{02}^{2}+
v_{02}^{2}\right)}
{3 z \left(2 u_{02}^{2}+v_{02}^{2}-v_{02}\right) } \\
   r_{11} =& \dfrac{v_{11} \left(4 \alpha u_{11}^{2} v_{11}+2 \alpha v_{11}-4 \alpha+2 u_{11}^{2} v_{11}+
v_{11}\right)}
{3 z \left(2 u_{11}^{2} v_{11}+v_{11}-1\right)} \\
r_{12} =&\dfrac{u_{12} v_{12} \left(2 \alpha u_{12} v_{12}^{2}+4 \alpha u_{12}-4 \alpha v_{12}+u_{12} v_{12}
^{2}+2 u_{12}\right)}{3 z \left(u_{12} v_{12}^{2}+2 u_{12}-v_{12}\right)} \end{align}
\begin{align}
r_{21}=&\dfrac{u_{21} v_{21}^{2} \left(2 \alpha u_{21} v_{21}^{2}+4 \alpha u_{21}-4 \alpha+u_{21} v_{21}
^{2}+2 u_{21}\right)}{3 z \left(u_{21} v_{21}^{2}+2 u_{21}-1\right)
}\,\\
r_{22}=&\dfrac{u_{22}^{2} v_{22} \left(2 \alpha u_{22}^{2} v_{22}^{2}-4 \alpha v_{22}+4 \alpha+u_{22}
^{2} v_{22}^{2}+2\right)}{3 z \left(u_{22}^{2} v_{22}^{2}-v_{22}+2
\right)}\,\\
r_{31}=&\left({6 z \left(2 u_{31} v_{31}^{2}+4 u_{31}+v_{31}\right)}\right)^{-1} \times \\ \nonumber
     & v_{31} \big(8 \alpha u_{31}^{2} v_{31}^{4}+16 \alpha u_{31}^{2} v_{31}^{2}+8 \alpha
 u_{31} v_{31}^{3} +2 \alpha v_{31}^{2}-4 \alpha+4 u_{31}^{2} v_{31}^{4}\\ \nonumber
	 &\hspace{1in}+8 u_{31}^{2} v_{31}^{2}+4 u_{31} v_{31}^{3}+8 u_{31} v_{31}+v_{31}^{2}+2\big)  \\
r_{32}=&u_{32} v_{32}^{2} \left({6 z \left(2 u_{32}^{2} v_{32}^{2}+u_{32} v_{32}^{2}+4\right)}\right)^{-1} \times \\ \nonumber
	& \big(8 \alpha u_{32}^{4} v_{32}^{2}+8 \alpha u_{32}^{3} v_{32}^{2}+
2 \alpha u_{32}^{2} v_{32}^{2}+16 \alpha u_{32}^{2}-4 \alpha+4 u_{32}^{4} v_{32}^{2}+4 u_{32}
^{3} v_{32}^{2}\\ \nonumber
	&\hspace{1in}+u_{32}^{2} v_{32}^{2}+8 u_{32}^{2}+8 u_{32}+2\big) 
\\
r_{42}=& v_{42}\left(6 z \left(2 u_{42}^{2} v_{42}^{2}+v_{42}+4\right)\right)^{-1} \times \\ \nonumber
	&\big(8 \alpha u_{42}^{6} v_{42}^{4}+8 \alpha u_{42}^{4} v_{42}^{3}+16 \alpha
 u_{42}^{4} v_{42}^{2}+2 \alpha u_{42}^{2} v_{42}^{2}-4 \alpha+4 u_{42}^{6} v_{42}^{4}+4 u_{42}
^{4} v_{42}^{3} \\ \nonumber
	& \hspace{1in}+8 u_{42}^{4} v_{42}^{2}+u_{42}^{2} v_{42}^{2}+8 u_{42}^{2} v_{42}+2\big) 
\\
r_{71}=&\dfrac{u_{71} \left(2 \alpha u_{71}^{2} v_{71}^{2}-4 \alpha u_{71} v_{71}^{2}+4 \alpha+u_{71}
^{2} v_{71}^{2}+2\right)}{3 z \left(u_{71}^{2} v_{71}^{2}-u_{71} v_{71}^{2}
+2\right)}\, .\label{lastr}
\end{align}

The part $L_0^{(9)}\backslash L_1^{(9)}$ is equal to the line $v_{02}=0$, on which $r_{02}=0$. The part $L_1^{(8)}\backslash L_2^{(7)}$ is equal to the line $u_{12}=0$, on which $r_{22}=0$. The part $L_3^{(6)}\backslash L_4^{(5)}$ is equal to the line $u_{32}=0$, on which $r_{32}=0$. The part $L_3^{(6)}\backslash L_2^{(7)}$ is equal to the line $v_{42}=0$ on which $r_{42}=0$. This covers the entire $I_3$ and so the proof for the first part of the lemma is complete.

For the second statement we have the line $L_4^{(5)}\backslash (L_3^{(6)}\cup L_5^{(4)})$ is equal to the line $v_{41}=0$, $u_{41}\neq 1/4$ on which $r_{41}=(1-2\alpha)(3z(4u_{41}+1))$. Therefore for every compact subset $K_1$ of  $L_4^{(5)}\backslash L_5^{(4)}$ there exists a neighbourhood $V_1$ such that $r_{41}$ is bounded. Similarly we have the line $L_7^{(2)}\backslash (L_0^{(9)}\cup L_8^{(1)})$ is equal to the line $u_{72}=0$, $v_{72}\neq 0$, on which $r_{72}=(1+2\alpha)/(3v_{72}z)$. Therefore on every compact subset $K_2$ of  $L_7^{(2)}\backslash L_8^{(1)}$ there exists a neighbourhood $V_2$ such that $r_{41}$ is bounded. We take $K=K_1 \cup K_2$ and $V=V_1\cup V_2$, which completes the proof of the second statement of the lemma.
\end{proof}

In the following two lemmas we require asymptotic information about the dynamics near the infinity set. From the appendix, we have that the line $L_{5}^{(4)}\backslash L_{4}^{(5)}$ is determined by the relations $v_{62}=0$, $u_{62} \in \mathbb{C}$, while the line $L_{8}^{(1)}\backslash L_{7}^{(2)}$ is determined by the relations $u_{91}=0$, $v_{91} \in \mathbb{C}$. Similarly the lines $L_{4}^{(5)}\backslash L_{3}^{(6)}$ and $L_{7}^{(2)}\backslash L_{0}^{(9)}$ are determined by $v_{52}=0$, $u_{52} \in \mathbb{C}$ and $u_{81}=0$, $v_{81} \in \mathbb{C}$ respectively. Table \ref{asymtable} shows the leading order behaviour of the solutions and the functions used as a measure of distance to the infinity set near these lines.

\begin{center}
\begin{table}
\begin{tabular}{|l|l|}
  \hline
  &\\
  Near $L_5^{(4)}$ as $v_{62}\rightarrow 0\,$: & Near $L_8^{(1)}$ as $u_{91}\rightarrow 0\,$: \\
  $\dot{u}_{62} \sim -v_{62}^{-1} $ &  $\dot{v}_{91} \sim u_{91}^{-1}  $  \\
  $w_{62} \sim -v_{62}/4 $ &  $w_{91} \sim -u_{91} $ \\
  $\dfrac{\dot{w}_{62}}{w_{62}}\sim \dfrac{4}{3z}+O(w_{62}) $ & $\dfrac{\dot{w}_{91}}{w_{91}}\sim \dfrac{4}{3z}+O(w_{91}) $ \\
  $Ew_{62} \sim 1 - \dfrac{2\alpha-1}{12u_{62}z}\, .$ & $Ew_{91}\sim 1 - \dfrac{2\alpha+1}{3v_{91}z}\, .$\\
  &\\
  \hline
  &\\
   Near $L_4^{(5)}$ as $v_{52}\rightarrow 0\,$: & Near $L_7^{(2)}$: as $u_{81}\rightarrow 0\,$:\\
  $\dot{u}_{52} \sim -2v_{52}^{-1} $ &  $\dot{v}_{81} \sim 2u_{81}^{-1}  $  \\
   $\dot{v}_{52} \sim u_{52}^{-1} $ &  $\dot{u}_{81} \sim -u_{v_{1}}^{-1}  $  \\
  $w_{62} \sim -u_{52}v_{52}^2/4 $ &  $w_{91} \sim -v_{81}u_{81}^2 $ \\
  $Ew_{62} \sim 1 $ & $Ew_{91}\sim 1  $\\
   $\dfrac{\dot{E}}{E}\sim -\dfrac{4}{3z}+\dfrac{1-2\alpha}{12u_{52}z}\sim -\dfrac{4}{3z}+\dfrac{1-2\alpha}{12z}\dot{v}_{52}\, .$ &  $\dfrac{\dot{E}}{E}\sim -\dfrac{4}{3z}+\dfrac{1+2\alpha}{3v_{81}z}\sim -\dfrac{4}{3z}+\dfrac{1+2\alpha}{12z}\dot{u}_{81}\, .$ \\
  &\\
  \hline
\end{tabular}
  \caption{Comparison of asymptotic behaviour near the infinity set}
  \label{asymtable}
\end{table}
\end{center}

\begin{remark}
Note that the behaviour of the Boutroux-Painlev\'e system near the pairs $L_5^{(4)}$ and $L_8^{(1)}$, $L_4^{(5)}$ and $L_7^{(2)}$ are equivalent to leading order, up to multiplicative constants. Due to this observation, the proofs of the following lemmas will only consider one of the two in each pair, while the proof near the other can be inferred from the similar behaviour of its pair.
\end{remark}

\begin{remark}
The solution cannot be simultaneously close to both $L_6^{(3)}$ and $L_9$. This can be seen by the coordinate  relation
\begin{displaymath}
v_{91} \propto v_{61}^{-3}\,\, \mbox{as $v_{61} \rightarrow 0$.}
\end{displaymath}
That is, as a solution moves towards one of the pole lines (given by $v_{61}=0$ or $v_{91}=0$), the solution becomes unboundedly distant from the other pole line.
\end{remark}

In the following lemma we show that the three functions $1/E$, $w_{62}$ and $w_{91}$ can be stitched together  to form a continuous distance function, which will be later used to show that the infinity set is repelling. 

\begin{lemma} \label{lemma2}
Suppose $z$ is bounded away from zero. There exists a continuous complex valued function $d$ on a neighbourhood of $I$ in $S_9$ such that $d=1/E$ in a neighbourhood of $I \backslash (L_{5}^{(4)} \bigcup L_8^{(1)})$ in $S_9$, $d=w_{62}$ in a neighbourhood of $L_{5}^{(4)}\backslash L_{4}^{(5)}$ in $S_9$, and $d=w_{91}$ in a neighbourhood of $L_{8}^{(1)}\backslash L_{7}^{(2)}$ in $S_9$. Moreover, $E\,d\rightarrow 1$, $d/w_{52}\rightarrow 1$ when approaching  $L_{5}^{(4)}\backslash L_{4}^{(5)}$, and $E\,d\rightarrow 1$, $d/w_{91}\rightarrow 1$ when approaching  $L_{8}^{(1)}\backslash L_{7}^{(2)}$.

If the solution at the complex time $z$ is sufficiently close to a point of $L_{5}^{(5)}\backslash L_{4}^{(6)}$ (resp. $L_{8}^{(1)}\backslash L_{7}^{(2)}$), then there exists a unique $\zeta\in \mathbb{C}$, such that $|z-\zeta|={\mathcal O}(|d(z)||u_{62}(z)|)$ (resp. ${\mathcal O}(|d(z)||v_{91}(z)|))$, where d(z) is small, $u_{62}(z)$ (resp. $v_{91}(z)$) is bounded and $u_{62}(\zeta)=0$ (resp. $v_{91}(\zeta)=0$). That is, the point $z=\zeta$ is a pole of the Boutroux-Painlev\'e system. For large finite $R_5$  $\in \mathbb{R}_{>0}$, the connected component of $\zeta$ in $\mathbb{C}$ of the set of all $z\in \mathbb{C}$ such that $|u_{62}(z)|\leq R_5$ is an approximate disc $D_5$ with centre at $\zeta$ and radius $\sim |\delta | R_5$, and $z\mapsto u_{62}(z)$ is a complex analytic diffeomorphism from $D_5$ onto $\{u_{62}\in \mathbb{C} \bigm | |u_{62}|\leq R_5\}$.

Depending on which pole line the solution is close to, we write $\delta:=d(\zeta)=w_{62}(\zeta)$ or $\delta:=d(\zeta)=w_{91}(\zeta)$. Then we have $d(z)/\delta \sim 1$ as $\delta \rightarrow 0$. 
\end{lemma}
\begin{proof}
Without loss of generality, we assume the solution is near the part of the infinity set with which the poles of negative residue are associated. From the explicit details presented in the appendix it follows that  $L_{5}^{(5)}\backslash L_{4}^{(6)}$ is determined by the equations $v_{62}=0$, $u_{62} \in \mathbb{C}$. Asymptotically for $v_{62}\rightarrow 0$, and bounded $u_{62}$ and $z^{-1}$, we have
\begin{subequations}
\begin{eqnarray}
\dot{u}_{62} & \sim &  -v_{62}^{-1}  \\
w_{62} & \sim & -v_{62}/4  \\
\dot{w}_{62}/w_{62} &=&\frac{4}{3z} + {\mathcal O}(v_{62})=\frac{4}{3z}+{\mathcal O}(w_{62}) \label{w62eqn} \\
Ew_{62} &\sim & 1+ (1-2\alpha)/(12u_{62}z)\label{Ew62eqn}
\end{eqnarray}
\end{subequations}
Provided the solution to the Boutroux-Painlev\'e system is close to $L_{5}^{(5)}\backslash L_{4}^{(6)}$ then \eqref{w62eqn} gives 
\begin{displaymath}
w_{62}=\left(\frac{z}{\zeta}\right)^{4/3}w_{62}(\zeta)(1+{o}(1))\,.
\end{displaymath}

If $|z-\zeta|\ll |\zeta|$ (iff $z/\zeta\sim 1$) then we have $w_{62}(z)\sim w_{62}(\zeta)$ and so $v_{62}$ is approximately constant ($v_{62}\sim -4w_{62}(\zeta)$) and so

\begin{displaymath}
u_{62}(z)\sim u_{62}(\zeta)-v_{62}^{-1}(\zeta)(z-\zeta).
\end{displaymath}
So $u_{62}$ fills an approximate disc, centred at $u_{62}(\zeta)$ with radius $\sim R$ if $z$ runs over an approximate disc of radius $|v_{62}(\zeta)| R$. If $|v_{62}(\zeta)| \ll 1/|\zeta|$, then the solution at $z$ in an approximate disc $D$, centred at $\zeta$ with radius $\sim |v_{62}|R$ has the properties that $v_{62}(z)/v_{62}(\zeta) \sim 1$ and $z\rightarrow u_{62}(z)$ is a complex analytic diffeomorphism from $D$ onto an approximate disc centred at $u_{62}(\zeta)$ with radius $\sim R$. If $R$ is chosen to be sufficiently large, we have $0\in u_{62}(D)$, that is, the solution to the Boutroux-Painlev\'e system has a pole at a unique point in $D$ (as $u_{62}=0$ corresponds to a pole with residue $-1$). We may shift the centre of the approximate disc so that $u_{62}(\zeta)=0$, that is, shift the disc to be centred at the pole point.

Provided $|z-\zeta |\ll |\zeta|$, we have $d(z)/d(\zeta)=w_{62}(z)/w_{62}(\zeta) \sim 1$, that is, $w_{62}(z)/\delta \sim1$ and so $u_{62}(z) \sim 2^{-2} \delta^{-1}(z-\zeta)$. Then for given $R_5\in\mathbb{R}_{>0}$, the equation $|u_{62}(z)|=R_5$ corresponds to $|z-\zeta|\sim 2^2|\delta |R_5$, which is still small compared to $|\zeta |$ if $|\delta |$ is sufficiently small. It follows the connected component of $D_5$ of the set of all $z\in\mathbb{C}$ such that $|u_{62}(z)|\leq R_5$ is an approximate disc with centre at $\zeta$ and radius $|\delta |R_5$, more precisely, $z\mapsto u_{62}(z)$ is a complex analytic diffeomorphism from $D_5$ onto $\{ u_{62}\in\mathbb{C} \big\| |u_{62}|\leq R_5 \}$, and that $d(z)/\delta\sim 1$ for all $z\in D_5$. 

$E$ has a pole at $z=\zeta$, but it follows from the relation \eqref{Ew62eqn} that $Ew_{62}\sim 1$ when $1\gg |z^{-1}u_{62}^{-1}|\sim |\zeta^{-1}\delta(z-\zeta)^{-1}|$, that is, when $|z-\zeta|\gg |\delta|/|\zeta|$. As the approximate radius of $D_5$ is $|\delta|R_5\gg|\delta|/|\zeta|$ as $R_5\gg1/|\zeta|$, we have $Ew_{62}\sim 1$ for $z\in D_5\backslash D_6$, where $D_6$ is a disc centred at $\zeta$ with small radius compared to the radius of $D_5$.

The set $L_{4}^{(6)}\backslash L_5^{(5)}$ is visible in the coordinate system $(u_{52},v_{52})$, where it corresponds to the equation $v_{52}=0$ and is parametrised by $u_{52}\in\mathbb{C}$. The set $L_5^{(5)}$ minus one point corresponds to $u_{52}=0$ and is parametrised by $v_{52}\in\mathbb{C}$. The equations that express $(u_{51},v_{51})$ and $(u_{52},v_{52})$ in terms of $(u_{41},v_{41})$ show
\begin{subequations}
\begin{eqnarray}
u_{52} & = & u_{51}v_{51}   \\
v_{52} & = & u_{51}^{-1} \label{coord52}  \\
u_{51}&=&u_{62}+(1-2\alpha)/(12z),
\end{eqnarray}
\end{subequations}
which implies that $u_{62}\rightarrow \infty$ if and only if $v_{52}\rightarrow 0$. That is, when the point near $L_5^{(5)}$ approaches the intersection point with $L_4^{(6)}$, then $Ew_{62}\rightarrow 1$.

As remarked earlier, analogous arguments to the above case can be made where  the solution of the Boutroux-Painlev\'e system is close to $L_{8}^{(1)}\backslash L_{7}^{(2)}$. This completes the proof of the lemma.
\end{proof}

\begin{lemma}\label{lemma3}
For large finite $R_{4} \in \mathbb{R}_{>0}$, the connected component of $\zeta\in\mathbb{C}$ of the set of all $z\in\mathbb{C}$ such that the solution at the complex time z is close to $L_4^{(5)}\backslash L_{3}^{(6)}$, with $|u_{52}(z)|\leq R_4$, but not close to $L_{5}^{(4)}$, is the complement of $D_{5}$ in an approximate disk $D_4$ with centre at $\zeta$ and radius $\sim |\delta R_4|^{1/2}$. For all $z\in D_4$, the largest approximate disc, we have $|z-\zeta | \ll |\zeta|$ and $d(z)/\delta\sim 1$

The analogous statement holds true in the 7-th and 8-th coordinate charts.
\end{lemma}

\begin{proof}
Asymptotically as $v_{52}\rightarrow 0$, and for bounded $u_{52}$ and $z^{-1}$, we have
\begin{subequations}
\begin{eqnarray}
\dot{u}_{52} & \sim & -2v_{52}^{-1}   \label{v52eqns1}\\
\dot{v}_{52} & \sim & u_{52}^{-1}   \\
w_{52}&\sim& -u_{52}v_{52}^2  \\
Ew_{52}&\sim& 1 \\
\dot{E}/E &\sim& -4/3z +(1-2\alpha)/(3u_{52}z) \sim -4/3z +(1-2\alpha)/(3z)\dot{v}_{52}\, \label{v52eqns2}. 
\end{eqnarray}
\end{subequations}
So 
\begin{eqnarray}
\nonumber &\log(E(z_1)/E(z_0))\sim \log(z_1/z_0)^{-4/3}+(1-2\alpha)/3(v_{52}(z_1)/z_1-v_{52}(z_0)/z_0)\\
&\quad +\int_{z_0}^{z_1} z^{-2}v_{52}(z) dz ,
\end{eqnarray}
and hence
\begin{displaymath}
E(z_1)/E(z_0)\sim (z_1/z_0)^{-4/3}(1+o(1)) 
\end{displaymath}
It follows that $E(z_1)/E(z_0)\sim 1$ if for all $z$ on the segment from $z_0$ to $z_1$ we have $|z/z_0|\sim 1$, that is, $|z-z_0|\ll |z_0|$ and $|v_{52}|\ll |z_0|$.

We chose $z_0$ on the boundary of $D_5$, when $\delta/d(z_0)\sim E(z_0)\delta\sim E(z_0)w_{62}(z_0)\sim 1$, and 

\begin{eqnarray}
|u_{62}(z_0)|=R_5& \implies & |\frac{2\alpha-1}{3z}+\frac{1}{v_{52}}|=R_5 \\
 & \implies & |v_{52}(z_0)| \sim R_5^{-1} \ll 1\,. 
\end{eqnarray}
Furthermore

\begin{displaymath}
|u_{52}(z_0)|\sim |w_{52}v_{52}^{-2}|\sim |\delta| R_5^{-2} 
\end{displaymath}
which is small when $\delta$ is sufficiently small. If we have $z=\zeta+r(z_0-\zeta)$, $r\geq1$ we have $|u_{62}(z)|\geq R_5 \gg 1$, and hence $|v_{52}|\ll1$ from \eqref{coord52}. We also have $|z-z_0|/|z_0|=(r-1)|1-\zeta/z_0|\ll1$ if $r-1$ is small. 

Now (\ref{v52eqns1})-(\ref{v52eqns2}) and $E\sim \delta^{-1}$ give $v_{52}^{-1}\sim (-u_{52}\delta^{-1})^{1/2}$, and so 

\begin{subequations}
\begin{eqnarray}
\frac{d(u_{52}^{1/2})}{dz} & \sim & -i\delta^{-1/2}  \\
\implies  u_{52}^{1/2}(z) & \sim & u_{52}^{1/2}(z_0) -i\delta^{-1/2}(z-z_0)   \\
u_{52}(z)&\sim&-\delta^{-1}(z-z_0)^2 
\end{eqnarray}
\end{subequations}
if $|\delta^{-1/2}||z-z_0|\gg u_{52}(z_0)^{1/2}$.

For large $R_4\in\mathbb{R}_{>0}$, the equation $|u_{52}(z)|=R_4$ corresponds to $|z-z_0|\sim |\delta R_4 |^{1/2}$, which is small compared to $|z_0|\sim |\zeta|$, and so
\begin{displaymath}
|z-\zeta|\leq|z-z_0|+|z_0-\zeta|\ll |\zeta|\, .
\end{displaymath}
\end{proof}
We are now in a position to prove Theorem \ref{theorepel}.
\begin{proof}[Proof of Theorem \ref{theorepel}]
Suppose a solution of the Boutroux-Painlev\'e system is near the set $I_3$ at times $z_0$ and $z_1$. It follows from Lemma \ref{lemma3} that we have for every solution close to $I$, the set of complex times $z$ such that the solution is not close to $I_3$ is the union of approximate discs of radius $\sim |d|^{1/2}$ contained within approximate discs of radius $\sim |d|^{1/3}$. It also follows from Lemma \ref{lemma3} that the annular region between these discs is at least of order $|d|^{1/3} - |d|^{1/2} \sim |d|^{1/3}$, as  $|d|^{1/3} \gg |d|^{1/2}$.

Hence there if the solution is near $I$ for all complex times $z$ such that $|z_0|\leq |z|\leq |z_1|$, then there exists a path $\gamma$ from $z_0$ to $z_1$ such that the solution is close to $I_3$ \emph{for all} $z\in \gamma$, and $\gamma$ is $\mathbf{C}^1$ close to the path $\left[0,1\right]\ni t \mapsto z_1^t/z_0^{1-t}$.

Then Lemma \ref{lemma1} implies that 

\begin{align}
\nonumber {\rm log}(E(z)/E(z_0))&=-\dfrac{4}{3}{\rm log}(z/z_0)\int_0^1 dt +o(1) \\
\Rightarrow\ E(z)&=E(z_0)(z/z_0)^{-4/3(1+o(1))}(1+o(1))  
\end{align}
and so we have

\begin{align}\label{dzcor1}
d(z)&=d(z_0)(z/z_0)^{4/3(1+o(1))}(1+o(1))\,.
\end{align} 
Then Lemma \ref{lemma2} implies that, as long as we are close to $I$, as the solution moves into the region where $d$ is given by one of the two Jacobians $w_{62}$ or $w_{92}$, the ratio of $d$ remains close to 1.

For the first statement of the theorem, we have 

\begin{align*}
   \delta \,> & \,d(z)\, \geq d(z_0)\left(\frac{z}{z_0}\right)^{4/3-\epsilon_2}(1-\epsilon_3) 
   \end{align*}
   and so
   \begin{align*}
   \delta \, \geq & \sup_{z|\,|d(z)|<\delta}  d(z_0)\left(\frac{z}{z_0}\right)^{4/3-\epsilon_2}(1-\epsilon_3) \,.
\end{align*}
The second statement follows directly from \eqref{dzcor1}, while the third follows by the assumption on $z$.
\end{proof}

\section{Conclusion}\label{conc}

In this paper we have constructed Okamoto's space of initial conditions for the Boutroux-Painlev\'e system describing the behaviour of the second and thirty-fourth Painlev\'e equations in the asymptotic limit where the independent variable goes to infinity. Not treated here, but of interest, is the case where the parameter $\alpha$ of $\Ptw$ goes to infinity, cf. \cite{joshi:99,kaw:05}. From our explicit construction, we are able to conclude that each respective limit set of solutions to these equations forms a compact, connected and non-empty set of the space $S_9(\infty) \backslash I(\infty)$, which is the elliptic surface corresponding to the space of initial conditions for the autonomous limit system. We also showed that all solutions except those which vanish uniformly at infinity necessarily have infinitely many poles. Moreover, the special solutions to $\Ptw$ were found at singular values of the energy function $E$, where the branch points of the elliptic curve coalesce.

\appendix 
\section{Resolution of singularities for (\ref{boutrouxsys})-(\ref{boutrouxsys2})} \label{app1}
In this appendix, we construct Okamoto's space of initial conditions for the Boutroux-Painlev\'e system explicitly. Recall the notation from Section \ref{setup}, where the coordinates $(u_{ij},v_{ij})$ refer to the two coordinates in the $j$-th patch of the $i$-th blow up, $j=1,2$, $i=0,1,...,9$. 

The resolution of the Boutroux-Painlev\'e system can be seen in Figure \ref{blowupfig}, and can be summarised by the following diagram, where we omit the coordinate charts which are free from base points:

\begin{align*}
(u_{02},v_{02})&= (u/v,1/v)\xleftarrow{(0,0)} (u_{12},v_{12})\xleftarrow{(0,0)} (u_{21},v_{21})    \\
   &\xleftarrow{(1/2,0)} (u_{31},v_{31}) \xleftarrow{(0,0)} (u_{41},v_{41}) \xleftarrow{(-1/4,0)}  (u_{51},v_{51}) \xleftarrow{(\frac{1-2\alpha}{12z},0)}  (u_{61},v_{61}),\\
   (u_{01},v_{01}) &= (1/u,v/u)  
   \xleftarrow{(0,0)} (u_{72},v_{72}) \xleftarrow{(0,0)} (u_{82},v_{82}) \xleftarrow{(0,\frac{1+2\alpha}{3z})}  (u_{91},v_{91}).\\ 
\end{align*}
Here the label above each arrow represents the base point that is blown up in the preceding coordinate chart.

\begin{center}
\begin{figure}[h!]

\scalebox{.9} 
{
\begin{pspicture}(0,-4.07)(16.06291,4.07)
\psline[linewidth=0.04cm](2.3810155,1.67)(0.78101563,-0.93)
\psline[linewidth=0.04cm](0.58101565,-0.73)(3.9810157,-0.73)
\psline[linewidth=0.04cm](2.1810157,1.67)(3.7810156,-0.93)
\psline[linewidth=0.04cm](5.9610157,3.45)(5.9610157,-2.55)
\psline[linewidth=0.04cm](5.5810156,-2.25)(9.621016,-2.25)
\psline[linewidth=0.04cm](8.741015,-2.45)(11.1610155,-1.47)
\psline[linewidth=0.04cm](10.801016,-1.79)(11.781015,-0.09)
\psline[linewidth=0.04cm](11.561016,-0.75)(11.561016,2.25)
\psline[linewidth=0.04cm](11.961016,1.65)(9.761016,3.65)
\psline[linewidth=0.04cm](5.5610156,3.25)(10.561016,3.25)
\psline[linewidth=0.04cm](10.382354,2.25)(12.168049,4.05)
\psline[linewidth=0.04cm](11.771228,4.05)(13.358512,2.45)
\psline[linewidth=0.04cm](12.961016,2.45)(14.361015,3.85)
\psline[linewidth=0.04cm,linestyle=dashed,dash=0.16cm 0.16cm](13.953743,3.85)(15.541027,2.05)
\psline[linewidth=0.04cm,linestyle=dashed,dash=0.16cm 0.16cm](9.961016,-1.49)(10.941015,-4.05)
\usefont{T1}{ptm}{m}{n}
\rput(0.2324707,-0.705){$u$}
\usefont{T1}{ptm}{m}{n}
\rput(0.7124707,-1.145){$v$}
\usefont{T1}{ptm}{m}{n}
\rput(4.392471,-0.745){$1/u$}
\usefont{T1}{ptm}{m}{n}
\rput(2.5324707,1.935){$1/v$}
\usefont{T1}{ptm}{m}{n}
\rput(1.8724707,1.955){$u/v$}
\usefont{T1}{ptm}{m}{n}
\rput(3.8324707,-1.165){$v/u$}
\usefont{T1}{ptm}{m}{n}
\rput(12.122471,0.775){$L_0^{(9)}$}
\usefont{T1}{ptm}{m}{n}
\rput(8.082471,2.935){$L_1^{(8)}$}
\usefont{T1}{ptm}{m}{n}
\rput(10.202471,2.675){$L_2^{(7)}$}
\usefont{T1}{ptm}{m}{n}
\rput(11.22247,3.535){$L_3^{(6)}$}
\usefont{T1}{ptm}{m}{n}
\rput(12.862471,3.475){$L_4^{(5)}$}
\usefont{T1}{ptm}{m}{n}
\rput(13.942471,2.895){$L_5^{(4)}$}
\usefont{T1}{ptm}{m}{n}
\rput(15.082471,3.215){$L_6^{(3)}$}
\usefont{T1}{ptm}{m}{n}
\rput(10.882471,-0.945){$L_7^{(2)}$}
\usefont{T1}{ptm}{m}{n}
\rput(9.422471,-1.825){$L_8^{(1)}$}
\usefont{T1}{ptm}{m}{n}
\rput(10.952471,-2.845){$L_9$}
\usefont{T1}{ptm}{m}{n}
\rput(6.7224708,0.535){$L:u=0$}
\usefont{T1}{ptm}{m}{n}
\rput(7.2824707,-2.445){$L:v=0$}
\psline[linewidth=0.04cm,arrowsize=0.05291667cm 2.0,arrowlength=1.4,arrowinset=0.4]{<-}(3.9610157,0.85)(5.5610156,0.85)
\usefont{T1}{ptm}{m}{n}
\rput(3.4524708,0.175){$L_0$}
\usefont{T1}{ptm}{m}{n}
\rput(2.2079492,-1.045){\color{red}1}
\usefont{T1}{ptm}{m}{n}
\rput(1.2079492,0.555){\color{red}1}
\usefont{T1}{ptm}{m}{n}
\rput(3.2079492,0.555){\color{red}1}
\usefont{T1}{ptm}{m}{n}
\rput(7.481748,-2.045){\color{red}-1}
\usefont{T1}{ptm}{m}{n}
\rput(5.638096,0.155){\color{red}0}
\usefont{T1}{ptm}{m}{n}
\rput(8.29625,3.555){\color{red}-2}
\usefont{T1}{ptm}{m}{n}
\rput(11.29625,2.555){\color{red}-2}
\usefont{T1}{ptm}{m}{n}
\rput(11.49625,3.155){\color{red}-2}
\usefont{T1}{ptm}{m}{n}
\rput(12.49625,2.955){\color{red}-2}
\usefont{T1}{ptm}{m}{n}
\rput(13.49625,3.355){\color{red}-2}
\usefont{T1}{ptm}{m}{n}
\rput(11.29625,0.755){\color{red}-2}
\usefont{T1}{ptm}{m}{n}
\rput(11.49625,-1.245){\color{red}-2}
\usefont{T1}{ptm}{m}{n}
\rput(10.49625,-1.445){\color{red}-2}
\usefont{T1}{ptm}{m}{n}
\rput(11.081748,-3.445){\color{red}-1}
\usefont{T1}{ptm}{m}{n}
\rput(15.481748,2.555){\color{red}-1}
\usefont{T1}{ptm}{m}{n}
\rput(4.771875,1.155){Blow up}
\end{pspicture} 
}

\caption{The 9 point blow up of $\mathbb{P}^2(\mathbb{C})$ showing the configuration of the exceptional curves. The numbers represent the self intersection of the lines they are adjacent to. The configuration of the irreducible divisors (the infinity set) is that of the root lattice $E_7^{(1)}$ (see Figure \ref{E71}). The dashed lines indicating $L_6^{(3)}$ and $L_9$ are the poles lines, where the vector field is transversal to the line and a crossing indicates a pole of residue $\pm 1$ for $u$.}
\label{blowupfig} 
\end{figure}
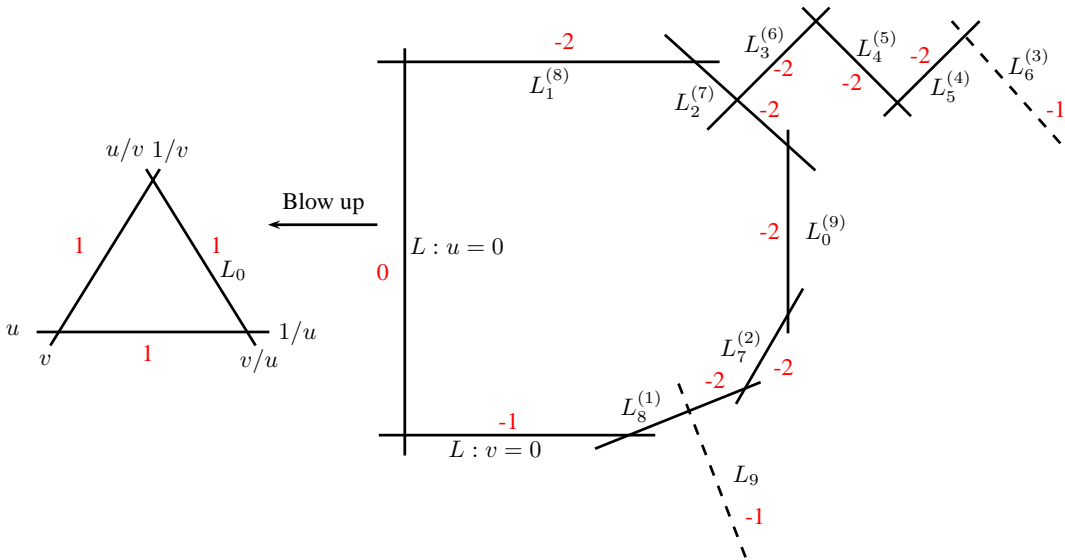
\end{center}

\begin{figure}
  \begin{center}
    \includegraphics{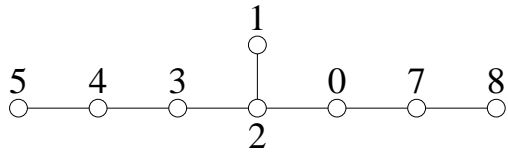}
  \end{center}
  \caption{\label{E71} 
           The Dynkin diagram for $E_{7}^{(1)}$, the numbers $i$ indicate the line $L_i$ which gives rise to the node. The nodes $j$ and $k$ are connected when $L_j^{(9-j)}$ intersects $L_k^{(9-k)}$.}
\end{figure}

\begin{remark}
The following blow up calculations are provided in explicit detail for completeness. The essential information for proofs in the body of the paper can be found in equations (\ref{firstr})-(\ref{lastr}) and Table \ref{asymtable}.
\end{remark}

\subsection{Embedding in $\mathbb{P}^2$}

In the second affine chart:
\begin{eqnarray*}
u_{01} & = & u^{-1} \\
v_{01} & = & vu^{-1}\\
\dot{u}_{01} &=&\frac{3 (u_{01}^{2}+2) z-2 (3 v_{01} z-1) u_{01}}
{6 z}\\
\dot{v}_{01} &=& \frac{2 (2 \alpha u_{01}^{2}+9 v_{01} z-(3 v_{01} z+1) u_{01} v_{01}
)+(3 v_{01} z+2) u_{01}^{2}}{6 u_{01} z}\\
E_{01}&=&\frac{-((u_{01}-v_{01}) u_{01}+2) v_{01}}{2 u_{01}^{3
}}\\
w_{01}&=&-u_{01}^{3}
\end{eqnarray*}

The line at infinity $L_0$ is given by $u_{01}=0$. There is a base point in this chart $b_7$ given by $u_{01}=0$, $v_{01}=0$ which is a base point of both the Painlev\'e vector field and of the autonomous energy function $E$.
In the third affine chart:
\begin{align*}
u_{02} & =  uv^{-1} \\
v_{02} & =  v^{-1}\\
\dot{u}_{02} &=\frac{-(3 ((v_{02}-2) v_{02}+6 u_{02}^{2}) z+4 \alpha 
u_{02} v_{02}^{2}+2 (v_{02}-1) u_{02} v_{02})}{6 v_{02} z
}\\
\dot{v}_{02} &=\frac{-((v_{02}-2) v_{02}+6 u_{02} z+2 \alpha v_{02}^{2})
}{3 z}\\
E_{02}&=\frac{-((v_{02}-1) v_{02}+2 u_{02}^{2})}{2 v_{02}
^{3}}\\
w_{02}&=-v_{02}^{3}
\end{align*}

The line at infinity $L_0$ is given by $v_{02}=0$. The vector field has a base point $b_0$ given by $u_{02}=0$, $v_{02}=0$, which is a base point of both the Painlev\'e vector field and of $E$.

\subsection{Resolution of $b_0$}

We blow up $\mathbb{P}^2$ at $b_0$. In the first chart:

\begin{eqnarray*}
u_{02} & = & u_{11} v_{11} \\
v_{02} & = & v_{11}\\
u&=&u_{11}\\
v&=&v_{11}^{-1}\\
\dot{u}_{11}&=&\frac{-(2 (3 u_{11} z+1) u_{11} v_{11}+3 (v_{11}-2) z
)}{6 v_{11} z}\\
\dot{v}_{11}&=&\frac{-(v_{11}-2+6 u_{11} z+2 \alpha v_{11}) v_{11}}{3 z} \\
w_{11}&=&-v_{11}^2\\
E_{11}&=&\frac{-(v_{11}-1+2 u_{11}^{2} v_{11})}{2 v_{11}^{2}}
\end{eqnarray*}

$v_{11}=0$ defines $L_1$ and the line $L_0^{(1)}$ is not visible in this chart. There are no base points in this chart.

In the second chart:

\begin{eqnarray*}
u_{02} & = & u_{12} \\
v_{02} & = & u_{12}v_{12} \\
u&=& v_{12}^{-1}\\
v&=& u_{12}^{-1} v_{12}^{-1}\\
\dot{u}_{12}&=&\frac{-(2 (u_{12}^{2} v_{12}-3 z+2 \alpha u_{12}^{2} v_{12}) v_{12}+
(3 v_{12}^{2} z-2 v_{12}+18 z) u_{12})}{6 v_{12} z}\\
\dot{v}_{12}&=&\frac{(3 v_{12}^{2} z+2 v_{12}+6 z) u_{12}-6 v_{12} z}{6 u_{12} 
z}\\
w_{12}&=&-u_{12}^{2} v_{12}^{3}\\
E_{12}&=&\frac{-u_{12} v_{12}^{2}-2 u_{12}+v_{12}}{2 u_{12}^{2} v_{12}^{3}}
\end{eqnarray*}

$u_{12}=0$ defines $L_1$ in this chart, and $v_{12}=0$ defines the lift of $L_0$=$L_0^{(1)}$. There is a base point $b_1$ of the vector field and of $E$ at $u_{12}=0$, $v_{12}=0$. 

\subsection{Resolution of $b_1$}

In the first chart:

\begin{eqnarray*}
u_{12} & = & u_{21}v_{21} \\
v_{12} & = & v_{21} \\
u&=&v_{21}^{-1}\\
v&=&u_{21}^{-1} v_{21}^{-2}\\
\dot{u}_{21}&=&\frac{-(u_{21}^{2} v_{21}^{3}-6 z+2 \alpha u_{21}^{2} v_{21}^{3}+3 (v_{21}^{2}
+4) u_{21} z)}{3 v_{21} z}\\
\dot{v}_{21}&=&\frac{(3 v_{21}^{2} z+2 v_{21}+6 z) u_{21}-6 z}{6 u_{21} z}\\
w_{21}&=&-u_{21}^{2} v_{21}^{4}\\
E_{21}&=&\frac{-(v_{21}^{2}+2) u_{21}+1}{2 u_{21}^{2} v_{21}^{4}}
\end{eqnarray*}

Here, $v_{21}=0$ defines $L_2$ and $u_{21}=0$ defines the proper transform $L_1^{(1)}$ of $L_1$. The proper transform $L_0^{(2)}$ of $L_0^{(1)}$ is not visible in this chart. There is a base point $b_2$ of both the vector field and the function $E$ given by $u_{21}=1/2$, $v_{21}=0$.

\begin{eqnarray*}
u_{12} & = & u_{22} \\
v_{12} & = & u_{22}v_{22} \\
u&=&u_{22}^{-1} v_{22}^{-1} \\
 v&=&u_{22}^{-2} v_{22}^{-1}\\
\dot{u}_{22}&=&\frac{-(2 u_{22}^{2} v_{22}+3 u_{22} v_{22} z-2) u_{22} v_{22}+6 (
v_{22}-3) z-4 \alpha u_{22}^{3} v_{22}^{2})}{6 v_{22} z}\\
 \dot{v}_{22}&=&\frac{(u_{22}+3 z) u_{22}^{2} v_{22}^{2}-6 (v_{22}-2) z+2 
\alpha u_{22}^{3} v_{22}^{2}}{3 u_{22} z}\\
w_{22}&=&-u_{22}^{4} v_{22}^{3}\\
E_{22}&=&\frac{-(u_{22}^{2} v_{22}-1) v_{22}-2}{2 u_{22}^{4} v_{22}^{3}}
\end{eqnarray*}

Here, $u_{22}=0$ defines $L_2$ and $v_{22}=0$ defines the proper transform $L_0^{(2)}$ of $L_0^{(1)}$. The proper transform $L_1^{(1)}$ of $L_1$ is not visible in this chart. There is a base point $b_2$ of both the vector field and the function $E$ given by $u_{22}=0$, $v_{22}=2$. 

\subsection{Resolution of $b_2$}

In the first chart we have

\begin{align*}
u_{21}=& u_{31} v_{31}+1/2 \\
 v_{21}=&v_{31}\\
 u=&v_{31}^{-1}\\ 
  v=&2/(2 u_{31} v_{31}^{3}+v_{31}^{2})\\
  \dot{u}_{31}=&-(12 (2 u_{31} v_{31}+1) v_{31} z)^{-1}\\
  \times &((8 u_{31}^{3} v_{31}^{4}+12 u_{31}^{2} v_{31}^{3}+36 u_{31}^{2} v_{31}^{
2} z+8 u_{31}^{2} v_{31}+120 u_{31}^{2} z+v_{31}+6 z) v_{31}\\
&+2 (3 v_{31}^{3}+15 v_{31}^{2} z+2 v_{31}+18 z) u_{31}
 +2 (2 u_{31} v_{31}+1)^{3} \alpha 
v_{31}^{2})\\
\dot{v}_{31}=&(6z (2 u_{31} v_{31}+1) )^{-1}((6 v_{31}^{2} z+4 v_{31}+12 z) u_{31} v_{31}+3 v_{31}^{2} z+2 v_{31}-6 z)\\
E=&(2u31 v31^{2}+4 u31+v31)^{-1}(v31^{3} (-4 u31^{2} v31^{2}-4 u31 v31-1))\\
w_{31}=&-2^{-2}(2 u_{31} v_{31}+1)^{2} v_{31}^{3}
\end{align*}

$v_{31}=0$ defines $L_3$, and $u_{31}v_{31}+1/2=0$ defines $L_1^{(2)}$, the proper transform of $L_1^{(1)}$. There is a base point $b_3$ in this chart given by $u_{31}=0$, $v_{31}=0$.

In the second chart:

\begin{align*}
u_{21}=&u_{32}+1/2\\
v_{21}=&u_{32} v_{32}\\
u=&u_{32}^{-1} v_{32}^{-1} \\
v=&2/(2 u_{32}^{3} v_{32}^{2}+u_{32}^{2} v_{32}^{2}) \\
\dot{u}_{32}=& -(12 v_{32} z)^{-1}\\
&\times(2 (2 u_{32}^{4} v_{32}^{3}+2 u_{32}^{3} v_{32}^{3}+3 u_{32} v_{32}^{2} z
+24 z)+(v_{32}+12 z) u_{32}^{2} v_{32}^{2}+2 (2 u_{32}+1
)^{2} \alpha u_{32}^{2} v_{32}^{3})			\\
\dot{v}_{32}=&(12 (2 u_{32}+1) u_{32} z)^{-1}\\
&\times (2 (2 (2 u_{32}^{5} v_{32}^{3}+3 u_{32}^{4} v_{32}^{3}+9 z)+3 
(v_{32}+6 z) u_{32}^{3} v_{32}^{2})+(v_{32}^{2}+30 v_{32} z+8
) u_{32}^{2} v_{32}\\
&+2 (3 v_{32}^{2} z+2 v_{32}+60 z) u_{32}
+2 (2 u_{32}+1)^{3} \alpha u_{32}^{2} v_{32}^{3})			\\
E=&	-((2 u_{32}+1)^{2} u_{32}^{3} v_{32}^{4})^{-1}(2 u_{32}^{2} v_{32}^{2}+u_{32} v_{32}^{2}+4))		\\
w_{32}=& -2^{-2}(2 u_{32}+1)^{2} u_{32}^{3} v_{32}^{4}		
\end{align*}

$u_{32}=0$, $u_{32}+1/2=0$ and $v_{32}=0$ define $L_3$, $L_1^{(2)}$ and $L_2^{(1)}$ respectively. There are no base points in this chart.

\subsection{Resolution of $b_3$}

In the first chart we have

\begin{align*}
u_{31}=&u_{41} v_{41}\\
v_{31}=&v_{41}\\
u=&v_{41}^{-1}\\
 v=& 2/(2 u_{41} v_{41}^{4}+v_{41}^{2})\\
\dot{u}_{41}=&-(12 (2 u_{41} v_{41}^{2}+1) v_{41} z)^{-1}\\
&\times (v_{41}+6 z+8 u_{41}^{3} v_{41}^{7}+4 (3 v_{41}^{3}+12 v_{41}^{2} z+4 
v_{41}+36 z) u_{41}^{2} v_{41}^{2}\\
&+2 (3 v_{41}^{3}+18 v_{41}^{2} z+4 v_{41}+12 
z) u_{41}+2 (2 u_{41} v_{41}^{2}+1)^{3} \alpha v_{41})  			\\
\dot{v}_{41}=&(6z (2 u_{41} v_{41}^{2}+1))^{-1}(2 (3 v_{41}^{2} z+2 v_{41}+6 z) u_{41} v_{41}^{2}+3 v_{41}^{2} z+2 v_{41}
-6 z)			\\
E=&-((2 u_{41} v_{41}^{2}+1)^{2} v_{41}^{2})^{-1}(2 u_{41} v_{41}^{2}+4 u_{41}+1)			\\
w_{41}=&-2^{-2}(2 u_{41} v_{41}^{2}+1)^{2} v_{41}^{2}
\end{align*}

$v_{41}=0$ defines $L_4$, and $u_{41}v_{41}^2+1/2=0$ defines $L_1^{(3)}$. There is a base point $b_4$ in this chart given by $u_{41}=-1/4$, $v_{41}=0$

In the second chart we have

\begin{align*}
u_{31}=&u_{42} \\
v_{31}=&u_{42} v_{42}\\
u=&u_{42}^{-1} v_{42}^{-1} \\
v=&2/(2 u_{42}^{4} v_{42}^{3}+u_{42}^{2} v_{42}^{2})\\
\dot{u}_{42}=&-(12 (2 u_{42}^{2} v_{42}+1) v_{42} z)^{-1}\\
&\times  (2 (2 (2 u_{42}^{3} v_{42}^{2}+3 u_{42} v_{42}+9 z) u_{42}
^{4} v_{42}^{3}+3 (v_{42}+6) z)+(30 u_{42} z+1) 
(v_{42}+4) u_{42} v_{42}\\
&+2 (3 v_{42}+4) u_{42}^{3} v_{42}^{2}
+2 
(2 u_{42}^{2} v_{42}+1)^{3} \alpha u_{42} v_{42}^{2})		\\
\dot{v}_{42}=&(12z (2 u_{42}^{2} v_{42}+1) u_{42} )^{-1}\\
&\times (2 (2 (2 u_{42}^{3} v_{42}^{2}+3 u_{42} v_{42}+12 z) u_{42}^{4} 
v_{42}^{3}+3 (v_{42}+4) z+18 (v_{42}+4) u_{42}^{2} v_{42} z
)\\
&+(v_{42}+8) u_{42} v_{42}+2 (3 v_{42}+8) u_{42}^{3} v_{42}^{2}
+2 (2 u_{42}^{2} v_{42}+1)^{3} \alpha u_{42} v_{42}^{2})		\\
E=&-((2 u_{42}^{2}
 v_{42}+1)^{2} u_{42}^{2} v_{42}^{3})^{-1}(2 u_{42}^{2} v_{42}^{2}+v_{42}+4))		\\
w_{42}=&-2^{-2}(2 u_{42}^{2} v_{42}+1)^{2} u_{42}^{2} v_{42}^{3}
\end{align*}

Here, $u_{42}=0$, $u_{42}^2v_{42}+1/2=0$ and $v_{42}=0$ define $L_4$, $L_1^{(3)}$ and $L_3^{(1)}$ respectively. There are no base points in this chart. There is a base point $b_4'$ in this chart given by $u_{42}=0$, $v_{42}=-4$, which is the manifestation of $b_4$ in this chart.

\subsection{Resolution of $b_4$}

In the first chart we have

\begin{align*}
u_{41}=& u_{51} v_{51}-1/4 \\
v_{41}=&v_{51}\\
u=&v_{51}^{-1}\\
 v=&4/(4u_{51}v_{51}^5-v_{51}^4+2v_{51}^2)   \\
\dot{u}_{51}=&-(48z (v_{51}^{2}-2-4 u_{51} v_{51}^{3}) v_{51} )^{-1}\\
&\times (v_{51}^{6}-6 v_{51}^{4}-24 v_{51}^{3} z+4 v_{51}^{2}+8-64 u_{51}^{3} v_{51}^{9}\\
&+48 (v_{51}^{5}-2 v_{51}^{3}-10 v_{51}^{2} z-4 v_{51}-28 z) u_{51}^{2} v_{51}
^{3}-2 (4 u_{51} v_{51}^{3}-v_{51}^{2}+2)^{3} \alpha\\
&-4 (3 v_{51}^{7}-12 v_{51}^{5}-54 v_{51}^{4} z-8 v_{51}^{3}-72 v_{51}^{2} z+24
 v_{51}+24 z) u_{51})		\\
\dot{v}_{51}=&(6z (v_{51}^{2}-2-4 u_{51} v_{51}^{3}))^{-1}(3 v_{51}^{4} z+2 v_{51}^{3}-4 v_{51}+12 z-4 (3 v_{51}^{2} z+2 v_{51}+6 z
) u_{51} v_{51}^{3})		\\
E=&-2 ((v_{51}^{2}-2-4 u_{51} v_{51}^{3})^{2} v_{51})^{-1} (4 u_{51} v_{51}^{2}+8 u_{51}-v_{51})		\\
w_{51}=&-2^{-4}(4 u_{51} v_{51}^{3}-v_{51}^{2}+2)^{2} v_{51}
\end{align*}

$v_{51}=0$ defines $L_5$, and $u_{51}v_{51}^3-v_{51}^2/4+1/2=0$ defines $L_1^{(4)}$. There is a base point $b_5$ in this chart given by $u_{51}=(1-2\alpha)/(12z)$, $v_{51}=0$

In the second chart we have

\begin{align*}
u_{41}=&u_{52}-1/4\\ 
v_{41}=&u_{52} v_{52}\\
u=&u_{52}^{-1} v_{52}^{-1} \\
 v=&4/(4 u_{52}^{5} v_{52}^{4}-u_{52}^{4} v_{52}^{4}+2
 u_{52}^{2} v_{52}^{2})\\
\dot{u}_{52}=&-(48 (4 u_{52}^{3} v_{52}^{2}-u_{52}^{2} v_{52}^{2}+2) v_{52} z)^{-1}\\
&\times (4 ((16 u_{52}^{8} v_{52}^{6}-12 u_{52}^{7} v_{52}^{6}+3 u_{52}^{
6} v_{52}^{6}+16) u_{52} v_{52}-2 (v_{52}-24 z)-12 (v_{52}-8 z
) u_{52}^{5} v_{52}^{4}\\
&-(v_{52}+72 z) u_{52}^{2} v_{52}^{2})
-(v_{52}^{2}-96) u_{52}^{6} v_{52}^{5}+2 (3 v_{52}^{2}-96 v_{52} z+64
) u_{52}^{4} v_{52}^{3}\\
&+8 (3 v_{52}^{2} z-2 v_{52}+144 z) u_{52}^{3} 
v_{52}^{2}+2 (4 u_{52}^{3} v_{52}^{2}-u_{52}^{2} v_{52}^{2}+2)^{3} \alpha v_{52}
)		\\
\dot{v}_{52}=&(48 (4 u_{52}^{3} v_{52}^{2}-u_{52}^{2} v_{52}^{2}+2) u_{52} z)^{-1}\\
&\times (4 ((16 u_{52}^{8} v_{52}^{6}-12 u_{52}^{7} v_{52}^{6}+3 u_{52}^{6} v_{52}
^{6}+24) u_{52} v_{52}-2 (v_{52}-12 z)-12 (v_{52}-10 z) 
u_{52}^{5} v_{52}^{4}\\
&-(v_{52}+72 z) u_{52}^{2} v_{52}^{2})-(v_{52}^{
2}-96) u_{52}^{6} v_{52}^{5}+6 (v_{52}^{2}-36 v_{52} z+32) u_{52}^{4} 
v_{52}^{3}\\
&+8 (3 v_{52}^{2} z-4 v_{52}+168 z) u_{52}^{3} v_{52}^{2}+2 (4
 u_{52}^{3} v_{52}^{2}-u_{52}^{2} v_{52}^{2}+2)^{3} \alpha v_{52})		\\
E=&-2((4 u_{52}^{3} v_{52}^{2}-u_{52}^{2} v_{52}^{2}+2)^{2} u_{52} v_{52}^{2})^{-1} (4 u_{52}^{2} v_{52}^{2}-u_{52} v_{52}^{2}+8)		\\
w_{52}=&-2^{-4}(4 u_{52}^{3} v_{52}^{2}-u_{52}^{2} v_{52}^{2}+2)^{2} u_{52} v_{52}^{2}
\end{align*}

Here, $u_{52}=0$, $u_{52}^3v_{52}^2-u_{52}^2v_{52}^2/4+1/2=0$ and $v_{52}=0$ define $L_5$, $L_1^{(4)}$ and $L_4^{(1)}$ respectively. There are no base points in this chart. There is a base point $b_5'$ in this chart given by $u_{52}=0$, $v_{52}=(12z)/(1-2\alpha)$, which is the manifestation of $b_5$ in this chart.

\subsection{Resolution of $b_5$}

In the first chart we have

\begin{align*}
u_{51}=& u_{61} v_{61} +(1-2 \alpha)/(12 z)\\
 v_{51}=&v_{61}\\
 u=&v_{61}^{-1}\\
  v=&12 z/((v_{61}^{3}-3 v_{61}^{2} z+6 z+12 u_{61}
 v_{61}^{4} z) v_{61}^{2}-2 \alpha v_{61}^{5})\\
\dot{u}_{61}=&-(432 (v_{61}^{3}-3 v_{61}^{2} z
+6 z+12 u_{61} v_{61}^{4} z-2 \alpha v_{61}^{3}) z^{3})^{-1}\\
&\times (v_{61}^{7}-9 v_{61}^{6} z+27 v_{61}^{5} z^{2}-18 v_{61}^{3} z^{2}-324 
v_{61}^{2} z^{3}-756 z^{3}+1728 u_{61}^{3} v_{61}^{10} z^{3}-16 \alpha^{4} v_{61}^{7}
\\
&-9 (3 z^{2}-2) v_{61}^{4} z +8 (2 v_{61}^{3}-9 v_{61}^{2} z+18 z+36 
u_{61} v_{61}^{4} z) \alpha^{3} v_{61}^{4}\\
&+72 (9 z^{2}+4) v_{61} z^{2}+
432 (v_{61}^{6}-3 v_{61}^{5} z+6 v_{61}^{3} z+36 v_{61}^{2} z^{2}+16 v_{61} z+96 z
^{2}) u_{61}^{2} v_{61}^{3} z^{2}\\
&+36 (v_{61}^{8}-6 v_{61}^{7} z+9 
v_{61}^{6} z^{2}+12 v_{61}^{5} z+30 v_{61}^{4} z^{2}+144 v_{61}^{2} z^{2}-216 v_{61} z
^{3}+96 z^{2}\\
&-4 (45 z^{2}-4) v_{61}^{3} z) u_{61} z\\
&+36 (
v_{61}^{5}-3 v_{61}^{4} z-2 v_{61}^{3}+22 v_{61}^{2} z+16 z-48 u_{61}^{2} v_{61}^{8} z\\
&-4 
(v_{61}^{3}-6 v_{61}^{2} z+12 z) u_{61} v_{61}^{4}) \alpha^{2} v_{61} z\\
&-2 (2 v_{61}^{7}-9 v_{61}^{6} z+180 v_{61}^{3} z^{2}-648 v_{61}^{2} z^{3}+432 
v_{61} z^{2}-324 z^{3}-1728 u_{61}^{3} v_{61}^{10} z^{3}\\
&+9 (3 z^{2}+2)v_{61}^{4} z +1296 (v_{61}^{2}-2) u_{61}^{2} v_{61}^{6} z^{3}\\
&+36 (v_{61}^{6}-9 v_{61}^{4} z^{2}+102 v_{61}^{2} z^{2}+16 v_{61} z+144 z^{2}) u_{61} v_{61}^{2} z) \alpha)		\\
\dot{v}_{61}=&(6 (v_{61}^{3}-3 v_{61}^{2} z+6 z+12 u_{61} v_{61}^{4} z-2 
\alpha v_{61}^{3}) z)	^{-1}\\
&\times(-2 v_{61}^{3} (\alpha-6 u_{61} v_{61} z) (3 v_{61}^{2} z+2 v_{61}+
6 z)+3 (v_{61}^{5}+4 v_{61}-12 z) z-(9 z^{2}-2) 
v_{61}^{4})	\\
E=&(v_{61} (v_{61}^{3}-3 v_{61}^{2} z+6 z+
12 u_{61} v_{61}^{4} z-2 \alpha v_{61}^{3})^{2})^{-1}\\
&\times(6 (2 \alpha v_{61}^{2}+4 \alpha-12 u_{61} v_{61}^{3} z-24 u_{61} v_{61} z-v_{61}
^{2}+3 v_{61} z-2) z)		\\
w_{61}=&-2^{-4}3^{-2} z^{-2}(2 \alpha v_{61}^{3}-12 u_{61} v_{61}^{4} z-v_{61}^{3}+3 v_{61}^{2} z-6 z
)^{2}
\end{align*}

$v_{61}=0$ defines $L_6$, and $u_{61}v_{61}^4-v_{61}^2/4+v_{61}(1-2\alpha)/(12z)+1/2=0$ defines $L_1^{(5)}$. There are no base points in this chart. The vector field is non-zero and transversal to the line $L_6$.

In the second chart

\begin{align*}
u_{51}=& u_{62} +(1-2 \alpha)/(12 z)\\
v_{51}=&u_{62} v_{62}\\
u=&u_{62}^{-1} v_{62}^{-1}\\
v=&(12 z)/((12 u_{62}^{4} 
v_{62}^{3} z+u_{62}^{3} v_{62}^{3}-3 u_{62}^{2} v_{62}^{2} z+6 z) u_{62}^{2} v_{62}^{2}
-2 \alpha u_{62}^{5} v_{62}^{5})\\
\dot{u}_{62}=&-(432 (12 u_{62}^{4} v_{62}^{
3} z+u_{62}^{3} v_{62}^{3}-3 u_{62}^{2} v_{62}^{2} z+6 z-2 \alpha u_{62}^{3} v_{62}^{3}
) v_{62} z^{3})^{-1}\\
&\times(4 (4 (27 (4 u_{62}^{11} v_{62}^{9} z+u_{62}^{10} v_{62}
^{9}+6 z^{2}) z^{2}-\alpha^{4} u_{62}^{8} v_{62}^{9})-81 (v_{62}-
16) u_{62}^{3} v_{62}^{3} z^{3})\\
&+(v_{62}-216 z^{2}) u_{62}^{8
} v_{62}^{8}
+36 (v_{62}-36 z^{2}) u_{62}^{9} v_{62}^{8} z-108 (7 v_{62}
-24) u_{62} v_{62} z^{3}\\
&+27 (v_{62}^{2}+16 v_{62}+480 z^{2}) u_{62}^{6
} v_{62}^{5} z^{2}
-9 (v_{62}^{2}-36 v_{62} z^{2}-288 z^{2}) u_{62}^{7} 
v_{62}^{6} z\\
&-18 (v_{62}^{2}+324 v_{62} z^{2}-24 v_{62}-2016 z^{2}) u_{62}^{4
} v_{62}^{3} z^{2}+9 (96 (v_{62}
+6) z^{2}-(3 z^{2}-2) v_{62}^{2}) u_{62}^{5} v_{62}^{4} z\\
&+8 (36 u_{62}^{4} v_{62}^{3} z+2 u_{62}^{3} v_{62}^{3}-9 u_{62}^{2}
 v_{62}^{2} z+18 z) \alpha^{3} u_{62}^{5} v_{62}^{6}+72 ((9 z^{2}+
4) v_{62}-108 z^{2}) u_{62}^{2} v_{62}^{2} z^{2}\\
&-36 (48 u_{62}^{8} v_{62}^{6} z+4 u_{62}^{7} v_{62}^{6}-24 u_{62}^{6} v_{62}^{5} z-u_{62}
^{5} v_{62}^{5}+2 u_{62}^{3} v_{62}^{3}-22 u_{62}^{2} v_{62}^{2} z-16 z\\
&+3 (v_{62}+16
) u_{62}^{4} v_{62}^{3} z) \alpha^{2} u_{62}^{2} v_{62}^{3} z+2 (2 
(864 u_{62}^{10} v_{62}^{7} z^{3}-u_{62}^{7} v_{62}^{7}-216 u_{62} v_{62} z^{2}+162 z
^{3}\\
&-18 (v_{62}+36 z^{2}) u_{62}^{8} v_{62}^{6} z+108 (3 v_{62}-22
) u_{62}^{2} v_{62} z^{3}-18 (5 v_{62}+12) u_{62}^{3} v_{62}^{2} z^{2}
)\\
&+9 (v_{62}^{2}+36 v_{62} z^{2}+288 z^{2}) u_{62}^{6} v_{62}^{4} z-
9 ((3 z^{2}+2) v_{62}+384 z^{2}) u_{62}^{4} v_{62}^{3} z
) \alpha u_{62} v_{62}^{2})		\\
\dot{v}_{62}=&(432 (12 u_{62}^{4} v_{62}^{3} z+u_{62}^{3} v_{62}^{3}-3 u_{62}^{2
} v_{62}^{2} z+6 z-2 \alpha u_{62}^{3} v_{62}^{3}) z^{3})^{-1}\\
&\times  ((4 (27 (4 (4 u_{62} z+1) u_{62}^{9} v_{62}^{8}-
(7 v_{62}-32) z) z^{2}-4 \alpha^{4} u_{62}^{7} v_{62}^{8}-81 
(v_{62}-16) u_{62}^{2} v_{62}^{2} z^{3})\\
&+(v_{62}-216 z^{2}
) u_{62}^{7} v_{62}^{7}+36 (v_{62}-36 z^{2}) u_{62}^{8} v_{62}^{7} z+
27 (v_{62}^{2}+16 v_{62}+576 z^{2}) u_{62}^{5} v_{62}^{4} z^{2}\\
&-9 (
v_{62}^{2}-36 v_{62} z^{2}-288 z^{2}) u_{62}^{6} v_{62}^{5} z-18 (v_{62}^{2}+360 v_{62} z^{2}-32 v_{62}-2304 z^{2}) u_{62}^{3} v_{62}^{2} z^{2}\\
&+8 (36
 u_{62}^{4} v_{62}^{3} z+2 u_{62}^{3} v_{62}^{3}-9 u_{62}^{2} v_{62}^{2} z+18 z) 
\alpha^{3} u_{62}^{4} v_{62}^{5}\\
&+72 ((9 z^{2}+4) v_{62}-108 z^{2}
) u_{62} v_{62} z^{2}+9 (24 (5 v_{62}+32) z^{2}-(3 z^{
2}-2) v_{62}^{2}) u_{62}^{4} v_{62}^{3} z\\
&-36 (48 u_{62}^{8} v_{62}^{6}
 z+4 u_{62}^{7} v_{62}^{6}-24 u_{62}^{6} v_{62}^{5} z-u_{62}^{5} v_{62}^{5}+2 u_{62}^{3} v_{62}
^{3}-22 u_{62}^{2} v_{62}^{2} z-16 z\\
&+3 (v_{62}+16) u_{62}^{4} v_{62}^{3} z
) \alpha^{2} u_{62} v_{62}^{2} z+2 (2 (864 u_{62}^{10} v_{62}^{7} z^{
3}-u_{62}^{7} v_{62}^{7}-216 u_{62} v_{62} z^{2}+162 z^{3}\\
&+324 (v_{62}-8) 
u_{62}^{2} v_{62} z^{3}-18 (v_{62}+36 z^{2}) u_{62}^{8} v_{62}^{6} z\\
&-18 (5 v_{62}+16) u_{62}^{3} v_{62}^{2} z^{2})+9 (v_{62}^{2}+36 
v_{62} z^{2}+288 z^{2}) u_{62}^{6} v_{62}^{4} z\\
&-9 ((3 z^{2}+2
) v_{62}+408 z^{2}) u_{62}^{4} v_{62}^{3} z) \alpha v_{62}) 
v_{62})		\\
E=&((12 u_{62}^{4} v_{62}^{3} z+u_{62}^{3} v_{62}^{3}-3 u_{62}^{2} v_{62}^{2} z+6 z-2 \alpha u_{62}^{3} v_{62}^{3})^{2} u_{62} v_{62})^{-1}	\\
&\times (6z (2 \alpha u_{62}^{2} v_{62}^{2}+4 \alpha-12 u_{62}^{3} v_{62}^{2} z-u_{62}^{
2} v_{62}^{2}+3 u_{62} v_{62} z-24 u_{62} z-2) )	\\
w_{62}=& -2^{-4}3^{-2} z^{-2}(2 \alpha u_{62}^{3} v_{62}^{3}-12 u_{62}^{4} v_{62}^{3} z-u_{62}^{3} v_{62}^{3}+3 u_{62}^{2} v_{62}^{2} z-6 z)^{2} v_{62}
\end{align*}

Here, $u_{62}=0$, $u_{62}^4v_{62}^3-u_{62}^2v_{62}^2/4+u_{62}^3v_{62}^3(1-2\alpha)/(12z)+1/2=0$ and $v_{62}=0$ define $L_6$, $L_1^{(5)}$ and $L_5^{(1)}$ respectively. There are no base points in this chart.

We have thus resolved all the singularities of the system which originated from the base point $b_0$. We return to the 01 coordinate chart to resolve the other base point.

\subsection{Resolution of $b_6$}

The base point $b_6$ is from the 01-chart.

In the first chart:
\begin{align*}
u_{01}=&u_{71} v_{71}\\
v_{01}=&v_{71}\\
u=&1/(u_{71} v_{71})\\
 v=&1/u_{71}\\
\dot{u}_{71}=&-(3 v_{71} z)^{-1}(u_{71}^{2} v_{71}-2 u_{71} v_{71}+6 z+2 \alpha u_{71}^{2} v_{71})		\\
\dot{v}_{71}=&(6 u_{71} z)^{-1}(2 (2 \alpha u_{71}^{2} v_{71}+9 z-(3 v_{71} z+1) u_{71} v_{71})+(3 v_{71} z+2) u_{71}^{2} v_{71})		\\
E=&-(2 u_{71}^{3} v_{71}^{2})^{-1}((u_{71}-1) u_{71} v_{71}^{2}+2)		\\
w_{71}=&-u_{71}^{3} v_{71}^{2}
\end{align*}

Here $v_{71}=0$ and $u_{71}=0$ define the lines $L_7$ and $L_0^{(7)}$ respectively. There are no base points in this chart.

\begin{align*}
u_{01}=&u_{72} \\
v_{01}=&u_{72} v_{72}\\
u=&1/u_{72} \\
 v=&v_{72}\\
\dot{u}_{72}=&(6 z)^{-1}(2 (u_{72}+3 z)-3 (2 v_{72}-1) u_{72}^{2} z)		\\
\dot{v}_{72}=&(3 u_{72} z)^{-1}(2 (\alpha u_{72}+3 v_{72} z)-(2 v_{72}-1) u_{72})
		\\
E=&(2 u_{72}^{2})^{-1}(((v_{72}-1) u_{72}^{2}-2) v_{72})		\\
w_{72}=& -(3 v_{72})^{-1}u_{72}(2 (\alpha u_{72}+3 v_{72}^{2})-(2 v_{72}-1) u_{72}
)
\end{align*}

Here $u_{72}=0$ defines the line $L_7$, the line $L_0^{(7)}$ is not visible in this chart. There is a base point $b_7$ of both the vector field and the function $E$ given by $u_{72}=0$, $v_{72}=0$.

\subsection{Resolution of $b_7$}

\begin{align*}
u_{72}=&u_{81} v_{81} \\
v_{72}=&v_{81}\\
u=&1/(u_{81} v_{81}) \\
v=&v_{81}\\
\dot{u}_{81}=&-(6 v_{81} z)^{-1}(2 (2 \alpha u_{81}+3 z)+3 (2 v_{81}-1) u_{81}^{
2} v_{81}^{2} z-2 (3 v_{81}-1) u_{81})		\\
\dot{v}_{81}=&(3 u_{81} z)^{-1}(2 (\alpha u_{81}+3 z)-(2 v_{81}-1) u_{81})		\\
E=&(2 u_{81}^{2} v_{81})^{-1}((v_{81}-1) u_{81}^{2} v_{81}^{2}-2)		\\
w_{81}=&v_{81}^{-1}(-2 (u_{81} v_{81}^{2}-1))
\end{align*}

Here $v_{81}=0$ and $u_{81}=0$ define the lines $L_8$ and $L_7^{(1)}$ respectively. There is a base point $b_{8}'$ of the vector field given by $u_{81}=-3z(1+2 \alpha)^{-1}$, $v_{82}=0$ which is not a base point of the energy function $E$.

In the second chart:

\begin{align*}
u_{72}=&u_{82}  \\
v_{72}=&u_{82} v_{82}\\
u=&1/u_{82} \\
 v=&u_{82} v_{82}\\
\dot{u}_{82}=&-(6 z)^{-1}(6 u_{82}^{3} v_{82} z-3 u_{82}^{2} z-2 u_{82}-6 z)		\\
\dot{v}_{82}=&(6 u_{82} z)^{-1}(3 (2 u_{82}^{2} v_{82} z-u_{82} z-2) u_{82} v_{82}+2 (3 v_{82} z+1
)+4 \alpha)		\\
E=&(2 u_{82})^{-1}(((u_{82} v_{82}-1) u_{82}^{2}-2) v_{82})		\\
w_{82}=&-u_{82}
\end{align*}

Here $u_{82}=0$ defines the line $L_8$, the line $L_7^{(1)}$ is not visible in this chart. There is a base point $b_8$ of the vector field given by $u_{82}=0$, $v_{82}=-(1+2 \alpha)(3 z)^{-1}$. There is a base point $b_{8,\infty}$ of the autonomous energy function $E$ at $u_{82}=0$, $v_{82}=0$.

\subsection{Resolution of $b_8$}

In the first chart:

\begin{align*}
u_{82}=&u_{91} v_{91} \\
 v_{82}=& v_{91} -(1+2 \alpha)(3 z)^{-1}\\
 u=&1/(u_{91} v_{91}) \\
 v=&((3 v_{91} z-1) u_{91} v_{91}-2 
\alpha u_{91} v_{91})/(3 z)\\
\dot{u}_{91}=&-(18 z^{2})^{-1} u_{91}((2 (3 v_{91} z-1) u_{91} v_{91}-3 z) (6 
v_{91} z-1) u_{91}+8 (\alpha^{2} u_{91}^{2} v_{91}-3 z)\\
&-2 (2 
(9 v_{91} z-2) u_{91} v_{91}-3 z) \alpha u_{91}))	\\
\dot{v}_{91}=&-(18 u_{91} z^{2})^{-1}\\
& \times(2 (9 (u_{91} v_{91}-z) z-4 \alpha^{2} u_{91}^{3} v_{91}
^{2})-(2 (3 v_{91} z-1) u_{91} v_{91}-3 z) (3 
v_{91} z-1) u_{91}^{2} v_{91}\\
&+2 (4 (3 v_{91} z-1) u_{91} v_{91}-3 z
) \alpha u_{91}^{2} v_{91})	\\
E=&-(18 u_{91} v_{91} z^{2})^{-1}((2 \alpha u_{91} v_{91}+3 z-(3 v_{91} z-1) u_{91} v_{91}
) u_{91}^{2} v_{91}^{2}+6 z) (3 v_{91} z-1-2 \alpha)	\\
w_{91}=&-u_{91}
\end{align*}

Here $v_{91}=0$ and $u_{91}=0$ define the lines $L_9$ and $L_8^{(1)}$ respectively. There are no base points in this chart.

\begin{align*}
u_{82}=&u_{92} \\ 
v_{82}=& u_{92} v_{92} -(1+2 \alpha)(3 z)^{-1}\\
u=&1/u_{92}\\
 v=&(3z)^{-1}((3 u_{92} v_{92} z-1) u_{92}-2 \alpha u_{92})\\
\dot{u}_{92}=&-(6z)^{-1}(6 u_{92}^{4} v_{92} z-2 u_{92}^{3}-3 u_{92}^{2} z-2 u_{92}-6 z-4 \alpha u_{92}
^{3})		\\
\dot{v}_{92}=&(18 z^{2})^{-1}(3 (6 (2 u_{92} v_{92} z-1) u_{92}^{2} v_{92}-(8 v_{92}-1)) z+8 \alpha^{2} u_{92}-2 (9 v_{92} z^{2}-1) u_{92}\\
&-2 
(18 u_{92}^{2} v_{92} z-4 u_{92}-3 z) \alpha)		\\
E=&-(18 u_{92} z^{2})^{-1}((2 \alpha u_{92}+3 z-(3 u_{92} v_{92} z-1) u_{92})
 u_{92}^{2}+6 z) (3 u_{92} v_{92} z-1-2 \alpha)		\\
w_{92}=& -1
\end{align*}

Here $u_{92}=0$ defines the line $L_9$, the line $L_8^{(1)}$ is not visible in this chart. There are no base points in this chart.

The Boutroux-Painlev\'e vector field is regular and non-zero on $L_9$, furthermore the vector field is transversal to $L_9$.

\end{document}